\documentclass[preprint,nocopyrightspace,9pt]{sigplanconf}
\usepackage{amssymb,amsthm,amsmath,stmaryrd,mathrsfs,galois}
\usepackage{algorithm,algpseudocode}
\usepackage{subfigure,wrapfig}
\usepackage{tikz}
\usepackage{makeidx}

\usetikzlibrary{positioning}
\usetikzlibrary{chains}
\usetikzlibrary{decorations}
\usetikzlibrary{calc}
\usetikzlibrary{matrix}
\usetikzlibrary{scopes}
\usetikzlibrary{decorations.pathmorphing}
\usetikzlibrary{decorations.shapes}
\usetikzlibrary{shapes}
\usetikzlibrary{shapes.symbols}
\usetikzlibrary{decorations.pathreplacing}

\tikzstyle{base} = [>=stealth,thick]

\newcommand{\oldtext}[1]{}

\renewcommand{\phi}{\varphi}

\DeclareFontFamily{U}  {MnSymbolC}{}
\DeclareFontShape{U}{MnSymbolC}{m}{n}{
    <-6>  MnSymbolC5
   <6-7>  MnSymbolC6
   <7-8>  MnSymbolC7
   <8-9>  MnSymbolC8
   <9-10> MnSymbolC9
  <10-12> MnSymbolC10
  <12->   MnSymbolC12}{}
\DeclareSymbolFont{MnSyC}{U}{MnSymbolC}{m}{n}
\DeclareMathSymbol{\righthalfcup}{\mathrel}{MnSyC}{184}
\DeclareMathSymbol{\ostar}{\mathrel}{MnSyC}{102}
\newtheorem{theorem}{Theorem}[section]
\newtheorem{lemma}[theorem]{Lemma}
\newtheorem{proposition}[theorem]{Proposition}

\newtheorem{definition}[theorem]{Definition}

\newcommand{\sem}[1]{\llbracket #1 \rrbracket}
\newcommand{\natsem}[1]{\llbracket #1 \rrbracket^{\mathscr{C}}}
\newcommand{\abssem}[1]{\llbracket #1 \rrbracket^{\mathscr{A}}}
\newcommand{\KA}{\textsf{KA}}  % Kleene algebra
\newcommand{\PKA}{\textsf{PKA}}  % Pre-Kleene algebra
\newcommand{\QPKA}{\textsf{QPKA}}  % Quantified Pre-Kleene algebra

\newcommand{\entry}{v^{\textsf{entry}}}
\newcommand{\exit}{v^{\textsf{exit}}}

\newcommand{\formula}[1]{\fbox{$#1$}}

\newcommand{\paths}[2]{\textsf{Paths}[#1,#2]} % Set of paths from #1 to #2
\newcommand{\pathexp}[2]{\mathbf{P}[#1,#2]} % Path expression for #1,#2
\newcommand{\soundrel}{\Vdash} % soundness relation
\newcommand{\sound}[2]{#1 \soundrel #2} % #2 is sound abstraction of #1
\newcommand{\regexp}{\textsf{RegExp}}
\newcommand{\Act}{\textsf{Act}}
\newcommand{\Exp}{\textsf{Exp}}
\newcommand{\Env}{\textsf{Env}}
\newcommand{\Var}{\textsf{Var}}
\newcommand{\act}{\textsf{act}}
\newcommand{\src}{\textsf{src}}
\newcommand{\tgt}{\textsf{tgt}}
\newcommand{\mul}{\odot}
\newcommand{\plus}{\oplus}
\renewcommand{\star}{\ostar}
\newcommand{\closure}{\star}
\newcommand{\powerset}[1]{2^{#1}}

\newcommand{\widen}{\triangledown}

\newcommand{\relational}{\mathscr{R}}
\newcommand{\relplus}{\plus^\relational}
\newcommand{\relmul}{\mul^\relational}
\newcommand{\relstar}{{\star^\relational}}

\newcommand{\path}{\Pi}

\newcommand{\logi}{\mathscr{L}}
\newcommand{\lplus}{\plus^\logi}
\newcommand{\lmul}{\mul^\logi}
\newcommand{\lstar}{{\star^\logi}}
\newcommand{\lsem}[1]{\sem{#1}^\logi}

\newcommand{\iinterp}[3]{#1(#2)[#3]}

\newcommand{\tuple}[1]{\langle #1 \rangle}

\newcommand{\false}{\textit{false}}
\newcommand{\true}{\textit{true}}

\theoremstyle{definition}

\newenvironment{example}[1][]{%
  \par%
  \pushQED{\qed}%
  \trivlist%
\item\ignorespaces%
  \refstepcounter{theorem}%
  \textbf{Example \arabic{section}.\arabic{theorem}}\ifx&#1& \else \;(#1) \fi \hspace{2pt}
}{%
  \popQED\endtrivlist%
}

\title{An Algebraic Framework For Compositional Program Analysis}
\authorinfo{Azadeh Farzan \and Zachary Kincaid}
           {University of Toronto}
           {azadeh,zkincaid@cs.toronto.edu}

\begin{document}

\maketitle

\begin{abstract}
The purpose of a program analysis is to compute an {\em abstract meaning} for a program which approximates its dynamic behaviour.  A {\em compositional} program analysis accomplishes this task with a divide-and-conquer strategy: the meaning of a program is computed by dividing it into sub-programs, computing their meaning, and then combining the results.  Compositional program analyses are desirable because they can yield scalable (and easily parallelizable) program analyses.

This paper presents algebraic framework for designing, implementing, and proving the correctness of compositional program analyses.  A program analysis in our framework defined by an algebraic structure equipped with {\em sequencing}, {\em choice}, and {\em iteration} operations.  From the analysis {\em design} perspective, a particularly interesting consequence of this is that the meaning of a loop is computed by applying the iteration operator to the loop body.  This style of compositional loop analysis can yield interesting ways of computing loop invariants that cannot be defined iteratively.  We identify a class of algorithms, the so-called {\em path-expression} algorithms \cite{Tarjan1981,Scholz2007}, which can be used to efficiently implement analyses in our framework.  Lastly, we develop a theory for {\em proving the correctness} of an analysis by establishing an approximation relationship between an algebra defining a concrete semantics and an algebra defining an analysis.
\end{abstract}

%% \category{CR-number}{subcategory}{third-level}

%% \terms
%% term1, term2

%% \keywords
%% keyword1, keyword2

\section{Introduction} \label{sec:intro}

Compositional program analyses compute an approximation of a program's
behaviours by breaking a program into sub-programs, analyzing them, and then
combining the results.  A classical example of compositionality is
interprocedural analyses based on the \emph{summarization} approach
\cite{Sharir1981}, in which a summary is computed for every procedure and then
used to interpret calls to that procedure when the program is analyzed.
Compositionality is interesting for two main reasons.  First is
computational efficiency: compositionality is crucial to building scalable
program analyses \cite{Cousot2002}, and can be easily parallelized.  Second,
compositionality opens the door to interesting ways of computing loop
invariants, based on the assumption that a summary of the loop body
is available \cite{Ammarguellat1990,Popeea2006,Kroening2008,Biallas2012}.
These compositional analyses are presented using ad-hoc, analysis-specific arguments, rather than in a generic framework.
This paper presents an algebraic framework for designing, implementing, and
proving the correctness of such compositional program analyses.

We begin by describing the \emph{iterative framework} for program analysis
\cite{Kildall1973,Cousot1977} to illustrate what we mean by a program analysis
{\em framework}, and to clarify what this paper provides an alternative to.
In the iterative framework, an analysis designed by providing an abstract
semantic domain $D$, a (complete) lattice defining the space of possible
program properties and a set of transfer functions which interpret the meaning
of program action as a function $D \rightarrow D$.  The correctness of an
analysis (with respect to a given concrete semantics) is proved by
establishing an approximation relation between the concrete domain and the
abstract domain (e.g., a Galois connection) and showing that the abstract
transfer functions approximate the concrete ones.  The result of the analysis
is computed by via a chaotic iteration algorithm (e.g., \cite{Bourdoncle1993}), which
repeatedly applies transfer functions until a fixed point is reached (possibly
using a widening operator to ensure convergence).

Let us contrast the iterative framework with the algebraic framework described
in this paper.  In the algebraic framework, an analysis designed by providing
an abstract semantic domain $K$, an algebraic structure (called a
\emph{Pre-Kleene algebra}) defining the space of possible program properties
which is equipped with \emph{sequencing}, \emph{choice}, and \emph{iteration}
operations; and a \emph{semantic function} which interprets the meaning of a
program action as an element of $K$.  The correctness of an analysis (with
respect to a given concrete semantics) is proved by establishing an
soundness relation between the concrete semantic algebra and the abstract
algebra.  The result of the analysis is computed via a path expression
algorithm \cite{Tarjan1981,Scholz2007}, which re-interprets a regular
expression (representing a set of program paths) using the operations of the
semantic algebra.

A summary of the difference between the our algebraic framework and the
iterative framework is presented below:
\begin{figure}[h]
  \begin{center}
  \begin{tabular}{|l|l|}
    \hline
    Iterative Framework & Algebraic Framework\\
    \hline
    (Complete) lattice & Pre-Kleene Algebra\\
    Abstract transformers & Semantic function\\
    Chaotic iteration algorithm & Path-expression algorithm\\
    Galois connection & Soundness relation\\
    \hline
  \end{tabular}
  \end{center}
  \caption{The algebraic framework by analogy}
\end{figure}

Some program analyses make \emph{essential} use of compositionality
\cite{Ammarguellat1990,Popeea2006,Kroening2008,Biallas2012}  and cannot easily be described
in the iterative framework.\footnote{Note that there are also iterative
  analyses which cannot be described in our framework - we are arguing an
  \emph{alternative} rather than a \emph{replacement} to the iterative
  framework.} The primary contribution of this paper is a framework in which  
  these (and similar) analyses can be described.  Of course, these analyses can be (and
indeed, have been) presented without the aid of a framework, which raises the
question: \emph{why do we need program analysis frameworks?}  The practical
benefits provided by our framework include the following:
\begin{itemize}
\item A means to study compositional analyses in an
  abstract setting.  This allows us to prove widely-applicable theoretical
  results (e.g., Sections ~\ref{sec:intra_abstract}
  and~\ref{sec:inter_abstract}), and to design re-usable program analysis
  techniques (e.g., consider the wide variety of abstract domain constructions
  \cite{Cousot1979} that have been developed in the iterative framework).
\item A clearly defined interface between the definition of a program analysis
  and the path-expression algorithms used to compute the result of an
  analysis.  This allows for generic implementations of (perhaps very
  sophisticated) path-expression algorithms to be developed independently of
  program analyses.
\item A conceptual foundation upon which to build compositional analyses.  For
  example, rather than ask an analysis designer to ``find a way to abstract
  loops,'' our framework poses the much more concrete, approachable problem
  of ``design an iteration operator.''
\end{itemize}

%% The remainder of the paper will proceed as follows.  We first give a
%% motivating example of an analysis which can be presented in our framework, but
%% not the iterative framework (Section~\ref{sec:motivation}).  We then present
%% our algebraic framework for compositional program analysis
%% (Sections~\ref{sec:interp} and~\ref{sec:intra_abstract}).  We adapt this
%% framework to accommodate programs with procedures, and provide an efficient
%% algorithm for solving interprocedural program analyses in the algebraic
%% framework (Sections~\ref{sec:interproc}, \ref{sec:inter_analysis}, and
%% \ref{sec:inter_abstract}).  Our interprocedural framework includes a novel
%% solution for handling local variables in interprocedural analyses, based on
%% the treatment of quantifiers in algebraic logic \cite{Henkin1971,Halmos1962}.
%% We present experimental results for our motivating example analysis in
%% Section~\ref{sec:experiments}, which indicate that our algebraic framework is
%% practical, and that compositional loop analysis is an interesting research
%% direction worthy of further exploration.  Finally, Section~\ref{sec:rel_work}
%% concludes and presents related work.

\noindent In summary, the contributions of this paper are:
\begin{itemize}
\item We develop a generic, algebraic framework for compositional program
  analyses.  This framework provides a unifying view for compositional program
  analyses, and in particular compositional loop analysis
  \cite{Ammarguellat1990,Popeea2006,Kroening2008,Biallas2012}.  Our work also
  unifies several uses of algebra in program analysis
  \cite{Reps2005,Kot2004,Bouajjani2003,Filipiuk2011}.
\item We develop a theory of abstract interpretation \cite{Cousot1977} that is
  specialized for algebraic analyses (Section~\ref{sec:intra_abstract}).  This
  allows an analysis designer to prove the correctness of an analysis by
  exhibiting a congruence relation between a concrete semantic algebra
  (defining a concrete semantics) and an abstract semantic algebra (defining
  an analysis).
\item We develop a variant of our algebraic framework that is suitable for
  interprocedural program analyses (Sections~\ref{sec:interproc}
  and~\ref{sec:inter_analysis}).  This framework includes a novel solution
  for handling local variables.
%% , based on the treatment of quantifiers in
%%   algebraic logic \cite{Henkin1971,Halmos1962}.
\item We prove a pair of precision theorems (one intraprocedural
  (Section~\ref{sec:intra_correct}) and one interprocedural
  (Section~\ref{sec:inter_correctness})) which identify conditions under which
  the path expression algorithms compute an ideal (merge-over-path) solution
  to a program analysis.
\item We develop a new program analysis, linear recurrence analysis, designed
  in our algebraic framework (Section~\ref{sec:motivation}).  We exhibit
  experimental results which demonstrate that compositional loop analyses
  (well-suited to our algebraic framework) have the potential to generate
  precise abstractions of loops.  The area of compositional loop analysis is
  relatively under-explored (as compared to iterative loop analysis), and our
  work provides a framework for exploring this line of research.
\end{itemize}

\section{Overview and Motivating Example} \label{sec:motivation}

The conceptual basis of our framework is Tarjan's \emph{path
  expression} algorithm \cite{Tarjan1981,Tarjan1981b}.  This algorithm takes two
parameters: a program, and an \emph{interpretation} (defining a program
analysis).  We may think of this algorithm as a two-step process.  (1) For
each vertex $v$, compute a path expression for $v$, which is is a regular
expression that represents the set of its control flow paths of the program
which end at $v$.  (2) For each vertex $v$, interpret the path expression for
$v$ (using the given interpretation) to compute an approximation of the
executions of the program which end at $v$.

Considering the program pictured in Figure~\ref{fig:example}, one possible
path expression for $v_8$ (the control location immediately before the
assertion) is as follows:
\begin{center}
\begin{minipage}{6cm}
$\tuple{\entry,v_1}\cdot\tuple{v_1,v_2}\cdot$\\
\hspace*{1cm}$\big(\tuple{v_2,v_3}\cdot\tuple{v_3,v_4}\cdot$\\
\hspace*{2cm}$(\tuple{v_4,v_5} \cdot \tuple{v_5,v_6} \cdot \tuple{v_6,v_4})^*$\\
\hspace*{1cm}$\cdot \tuple{v_4,v_7}\cdot\tuple{v_7,v_2}\big)^*$\\
$\cdot\tuple{v_2,v_8}$
\end{minipage}
\end{center}

An \emph{interpretation} consists of a \emph{semantic algebra} and a
\emph{semantic function}.  The semantic algebra consists of a \emph{universe}
which defines the space of possible program meanings, and \emph{sequencing},
\emph{choice}, and \emph{iteration} operators, which define how to compose
program meanings.  The semantic function is a mapping from control flow edges
to elements of the universe which defines the meaning of each control flow
edge.  The abstract meaning of a program is obtained by recursion on the path
expression, by re-interpreting the regular operators with their counterparts
in the semantic algebra.

\subsection*{A Compositional Program Analysis for Loops}

We sketch a compositional program analysis which can be defined and proved
correct in our framework. This program analysis, which we call \emph{linear recurrence
  analysis}, interprets the meaning of a program to be an arithmetic formula
which represents its input/output behaviour. 
The analysis makes use of the power of \emph{Satisfiability Modulo Theories} (SMT)
solvers to synthesize interesting loop invariants.

This example serves two purposes.  First, it motivates the development of our
framework.  The strategy this analysis uses to compute an invariant for a loop
relies crucially on having access to a summary of the loop body (i.e., it
relies on the \emph{compositionality} assumption of our framework); this makes
it difficult to define this analysis in the typical iterative program analysis
framework.  Second, this analysis will be used as a running example to explain
the formalism in the remainder of the paper.  The analysis has not yet
appeared elsewhere in the literature, but it is bears conceptual similarity to
``loop leaping'' analyses \cite{Kroening2008,Biallas2012}, and the
compositional induction variable analysis in \cite{Ammarguellat1990}.

\vspace{-4pt}
\paragraph{Universe.} The universe of this semantic algebra, $L$, is the set of arithmetic formulae with free variables in $\Var
\cup \Var'$ (quotiented by logical equivalence) where $\Var$ is the set of
program variables (inputs), and $\Var'$ is a set of ``primed'' copies of
program variables (outputs).  Informally, each such formula represents a
transition relation between states, and we say that a formula $\varphi_P$
approximates a program fragment $P$ if all the input/output pairs of $P$
belong to the transition relation of $\varphi_P$.

%% if for any execution of $P$ starting in a state $\rho$
%% and ending in a state $\rho'$, the structure that maps each unprimed variable
%% $x$ to $\rho(x)$ and each primed variable $x'$ to $\rho'(x)$ is a model for
%% $\varphi_P$.

\vspace{-4pt}
\paragraph{ Semantic Function.} The semantic function of linear recurrence analysis is defined in the obvious
way (by considering the effect of the statement label of each edge).  For example (again, considering Figure~\ref{fig:example}),
\begin{align*}
  \lsem{\tuple{\entry,v_1}} &= \formula{q' = q \land r' = x \land t'=t \land x' = x \land y' = y}\\
  \lsem{\tuple{v_1,v_2}} &= \formula{q' = 0 \land r' = r \land t'=t \land x'=x \land y'=y}\\
  \lsem{\tuple{v_2,v_3}} &= \formula{\parbox{0.65\linewidth}{$r > y \land q'=q \land r'=r \land t'=t \land x'=x$\\$\land y'=y$}}\\
  &\hspace*{5pt}\vdots
\end{align*}

\vspace{-4pt}
\paragraph{Operators.} Formally, the sequencing and choice operators are defined as follows:
\begin{align*}
  \varphi \lmul \psi &= \exists x''. \varphi[x''/x'] \land \psi[x''/x] & \text{Sequencing}\\
  \varphi \lplus \psi &= \varphi \lor \psi & \text{Choice}
\end{align*}
(where $\varphi[x''/x']$ denotes $\varphi$ with each primed variable $x'$
replaced by its double-primed counterpart $x''$, and $\psi[x''/x]$ similarly
replaces unprimed variables with double-primed variables).

The semantic function, sequencing, and choice operators are enough to analyze
loop-free code.  Interestingly, since these operations are precise, loop-free code is analyzed \emph{without loss of
  information}.  We can illustrate how loop free code is analyzed by
considering how the meaning of the body of the inner-most loop is
computed:

\vspace*{2pt}\noindent$\lsem{\tuple{v_4,v_5}\cdot \tuple{v_5,v_6}} = \lsem{\tuple{v_4,v_5}} \lmul \lsem{\tuple{v_5,v_6}}$\\
\vspace*{2pt}\hspace*{0pt}\hfill$=\formula{t > 0 \land q'=q \land r'=r+1 \land t'=t \land x'=x \land y'=y}$

\vspace*{4pt}\noindent$\lsem{\tuple{v_4,v_5} \cdot \tuple{v_5,v_6} \cdot \tuple{v_6,v_4}}$\\
\vspace*{2pt}\hspace*{0pt}\hfill$=\lsem{\tuple{v_4,v_5} \!\cdot\! \tuple{v_5,v_6}} \lmul \lsem{\tuple{v_6,v_4}}$\\
\vspace*{2pt}\noindent\vspace*{2pt}\hspace*{-2pt}\hfill=\fbox{$t > 0 \land q'=q \land r'=r-1 \land t'=t-1 \land x'=x \land y'=y$}

\noindent In the following, we will refer to the preceding formula as
$\varphi_{\textsf{inner}}$.

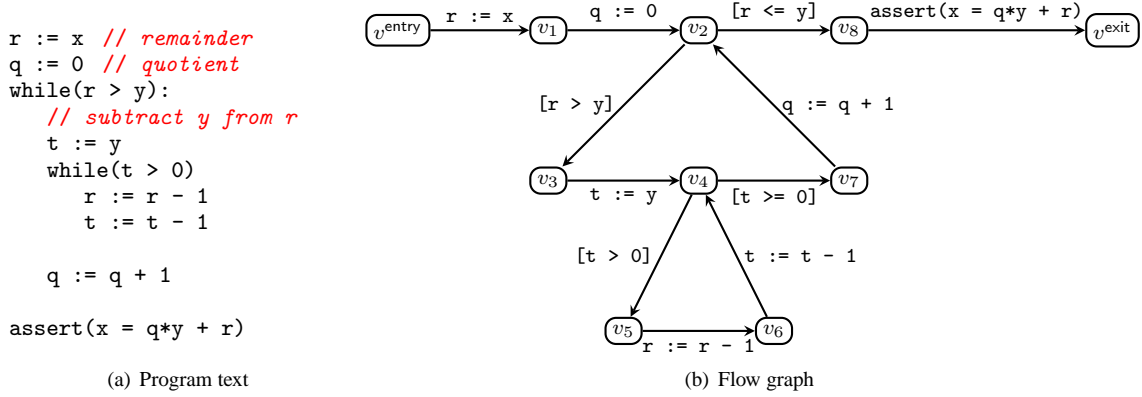
\begin{figure*}
  \begin{center}
  \subfigure[Program text]{
    \begin{minipage}[b]{4.5cm}
    %% \texttt{x := y := z := 0}\\
    %% \texttt{while(x < 10):}\\
    %% \hspace*{0.5cm}\texttt{x := x + 1}\\
    %% \hspace*{0.5cm}\texttt{y := y + 2}\\
    %% \hspace*{0.5cm}\texttt{if(y > 3):}\\
    %% \hspace*{1cm}\texttt{z := z + 3}\\
    %% \texttt{return}
    \texttt{r := x } \texttt{\color{red}// \textit{remainder}}\\
    \texttt{q := 0 } \texttt{\color{red}// \textit{quotient}}\\
    \texttt{while(r > y):}\\
    \hspace*{0.5cm}\texttt{\color{red}// \textit{subtract y from r}}\\
    \hspace*{0.5cm}\texttt{t := y}\\
    \hspace*{0.5cm}\texttt{while(t > 0)}\\
    \hspace*{1cm}\texttt{r := r - 1}\\
    \hspace*{1cm}\texttt{t := t - 1}\\\\
    \hspace*{0.5cm}\texttt{q := q + 1}\\\\
    \texttt{assert(x = q*y + r)}
    \vspace*{0.25cm}
    \end{minipage}
  }
  \subfigure[Flow graph]{
  \begin{tikzpicture}[base,node distance=2cm]
  \footnotesize
  \node [rectangle,rounded corners,draw] (entry) {$\entry$};
  \node[rectangle,rounded corners,draw,right of=entry] (v1) {$v_1$};
  \node[rectangle,rounded corners,draw,right of=v1] (v2) {$v_2$};
  \node[rectangle,rounded corners,draw,below of=v2,xshift=-2cm] (v3) {$v_3$};
  \node[rectangle,rounded corners,draw,right of=v3] (v4) {$v_4$};
  \node[rectangle,rounded corners,draw,below of=v4,xshift=-1cm] (v5) {$v_5$};
  \node[rectangle,rounded corners,draw,right of=v5] (v6) {$v_6$};
  \node[rectangle,rounded corners,draw,right of=v4] (v7) {$v_7$};
  \node[rectangle,rounded corners,draw,right of=v2] (v8) {$v_8$};
  \node[rectangle,rounded corners,draw,right of=v8,xshift=1.5cm] (exit) {$\exit$};

  \path (entry) edge[->] node[above]{\texttt{r := x}} (v1);
  \path (v1) edge[->] node[above]{\texttt{q := 0}} (v2);
  \path (v2) edge[->] node[left]{\texttt{[r > y]}} (v3);
  \path (v3) edge[->] node[below]{\texttt{t := y}} (v4);
  \path (v4) edge[->] node[below]{\texttt{[t >= 0]}} (v7);
  \path (v4) edge[->] node[left]{\texttt{[t > 0]}} (v5);
  \path (v5) edge[->] node[below]{\texttt{r := r - 1}} (v6);
  \path (v6) edge[->] node[right]{\texttt{t := t - 1}} (v4);
  \path (v7) edge[->] node[right]{\texttt{q := q + 1}} (v2);
  \path (v2) edge[->] node[above]{\texttt{[r <= y]}} (v8);
  \path (v8) edge[->] node[above]{\texttt{assert(x = q*y + r)}} (exit);

  %% \path (v2) edge[->] node[above]{\texttt{[r <= y]}} (v5);
  %% \path (v5) edge[->] node[above]{\texttt{assert(x = q*y + r)}} (v6);
  %% \path[sloped] (v4) edge[->] node[above]{\texttt{[y <= 3]}} (v1);
  %% \path (v4) edge[->] node[below]{\texttt{[y > 3]}} (v5);
  %% \path[sloped] (v5) edge[->] node[below]{\texttt{z := z + 3}} (v1);
  %% \path (v1) edge[->] node[left]{\texttt{[x >= 10]}} (v6);
  %% \path (v6) edge[->] node[left]{\texttt{return}} (exit);
  \end{tikzpicture}}
  %% \subfigure[Flow graph]{
  %% \begin{tikzpicture}[base,node distance=1.5cm,]
  %% \footnotesize
  %% \node [rectangle,rounded corners,draw] (entry) {$\entry$};
  %% \node[rectangle,rounded corners,draw,below of=entry] (v1) {$v_1$};
  %% \node[rectangle,rounded corners,draw,right of=v1,xshift=1cm] (v2) {$v_2$};
  %% \node[rectangle,rounded corners,draw,right of=v2,xshift=1cm] (v3) {$v_3$};
  %% \node[rectangle,rounded corners,draw,below of=v3] (v4) {$v_4$};
  %% \node[rectangle,rounded corners,draw,left of=v4,xshift=-1cm] (v5) {$v_5$};
  %% \node[rectangle,rounded corners,draw,below of=v1] (v6) {$v_6$};
  %% \node[rectangle,rounded corners,draw,below of=v6] (exit) {$\exit$};
  %% \path (entry) edge[->] node[left]{\texttt{x := y := z := 0}} (v1);
  %% \path (v1) edge[->] node[above]{\texttt{[x < 10]}} (v2);
  %% \path (v2) edge[->] node[above]{\texttt{x := x + 1}} (v3);
  %% \path[sloped] (v3) edge[->] node[above,yshift=0.15cm]{\texttt{y := y + 2}} (v4);
  %% \path[sloped] (v4) edge[->] node[above]{\texttt{[y <= 3]}} (v1);
  %% \path (v4) edge[->] node[below]{\texttt{[y > 3]}} (v5);
  %% \path[sloped] (v5) edge[->] node[below]{\texttt{z := z + 3}} (v1);
  %% \path (v1) edge[->] node[left]{\texttt{[x >= 10]}} (v6);
  %% \path (v6) edge[->] node[left]{\texttt{return}} (exit);
  %% \end{tikzpicture}}

  \end{center}
  \caption{An integer division program, computing a quotient and remainder} \label{fig:example}
\end{figure*}

The final step in defining our analysis is to provide a definition of the
iteration operator ($\lstar$).  The idea behind the iteration operator of
linear recurrence analysis is to use an SMT solver to extract linear
recurrence relations, and then solve the recurrence relations to compute
closed forms for (some of) the program variables.

Let $\phi$ be a formula (representing the body of a loop).  The
\emph{stratified linear induction variables} of $\phi$ are defined as follows:
\begin{itemize}
\item If $\phi \models x' = x + n$ for some constant $n \in \mathbb{Z}$, then
  $x$ is a linear induction variable of stratum 0.
\item If $\phi \models x' = x + f(\vec{y})$, where $f$ is an affine function
  whose variables range over linear induction variables of strictly lower
  strata than $n$, then $x$ is a linear induction variable of stratum $n$.
\end{itemize}

For example, in the formula $\phi_{\textsf{inner}}$, there are five induction
variables (all at stratum 0) satisfying the following recurrences:
\begin{center}
\begin{tabular}{|l|l|}
  \hline
  Recurrence & Closed form\\\hline
  $q' = q + 0$ & $q_k = q$\\
  $r' = r + 1$ & $r_k = r + k$\\
  $t' = t - 1$ & $t_k = t - k$\\
  $x' = x + 0$ & $x_k = x$\\
  $y' = y + 0$ & $y_k = y$\\ \hline
\end{tabular}
\end{center}

Stratified induction variables can be detected automatically using an SMT
solver.  We use $I(\varphi)$ to denote the set of all (stratified) induction
variables in $\varphi$.  For each $x \in I(\varphi)$, we may solve the
recurrence associated with $x$ to get a closed form, which is a function $cl_x$
with variables ranging over the program variables and a distinguished variable
$k$ indicating the iteration of the loop.  Then $\varphi^\lstar$ is defined as
follows:
\[
   \varphi^\lstar = \formula{\exists k. k \geq 0 \land \bigwedge_{x \in I(\phi)} x' = cl_x}
\]

In the particular case of $\phi_{\textsf{inner}}$, we have:
\begin{align*}
  \varphi_{\textsf{inner}}^\lstar &=  \fbox{\parbox{0.7\linewidth}{$\exists k. k \geq 0 \land q' = q \land r' = r - k \land t' = t - k$\\$\land x' = x \land y' = y$}}\\
  &= \formula{q' = q \land r' = r + t' - t \land t' \leq t \land x' = x \land y' = y}
\end{align*}

Describing how induction variables are detected and how recurrence relations
are solved is beyond the scope of this paper (which is about the program analysis \emph{framework} used for presenting this analysis, rather than the analysis itself).

\vspace{-4pt}
\paragraph{Computing the Meaning of the Program.}  Using the summary for the inner loop computed previously, we may proceed to compute
the abstract meaning of the body of the outer loop as:
\[
\varphi_{\textsf{outer}} = \formula{q' = q + 1 \land r' = r - y \land x' = x \land y' = y}
\]

\noindent We apply the iteration operator to compute the meaning of the outer loop.  We
detect four induction variables, three of which ($x,y,q$) are at stratum 0,
and one of which ($r$) is at stratum 1:
\begin{center}
\begin{tabular}{|l|l|}
  \hline
  Recurrence & Closed form\\\hline
  $q' = q + 1$ & $q_k = q + k$\\
  $r' = r - y$ & $r_k = r - yk$\\
  $x' = x + 0$ & $x_k = x$\\
  $y' = y + 0$ & $y_k = y$\\ \hline
\end{tabular}
\end{center}

We compute the abstract meaning of the outer loop as:
\begin{align*}
  \varphi_{\textsf{outer}}^\lstar &=  \fbox{\parbox{0.825\linewidth}{$\exists k. k \geq 0 \land q' = q + k \land r' = r - ky \land x' = x \land y' = y$}}\\
  &= \formula{q' \geq q \land r' = r - (q-q')y \land x' = x \land y' = y}
\end{align*}

Note the role of compositionality in computing the summary for the outer loop.
Detection of stratified induction variables depends only on having access to a
formula representing a loop body.  \emph{The fact that the outer loop body contains
another loop is completely transparent to the analysis.}

Finally, an input/output formula for the entire program is computed:
\begin{align*}
  \varphi_P &= \parbox{0.9\linewidth}{$\lsem{\tuple{\entry,v_1}\cdot\tuple{v_1,v_2}}\lmul \varphi_{\textsf{outer}}^\lstar \lmul \lsem{\tuple{v_2,v_8} \cdot \tuple{v_8,\exit}}$}\\
  &= \formula{q' \geq 0 \land r' = x - q'y \land r \leq y \land x' = x \land y'=y}
\end{align*}

This formula is strong enough to imply that the assertion at exit (the
standard post-condition for the quotient/remainder computation) correct.  This
is particularly interesting because it requires proving a \emph{non-linear}
loop invariant, which is out of scope for many state-of-the-art program
analyzers (e.g., \cite{ufo,invgen}).

 \section{Preliminaries}

 In this section, we will introduce our program model and review the concept
 of \emph{path expressions} \cite{Tarjan1981,Tarjan1981b}, which is the basis
 of our program analysis framework.

%% We will also provide a
%%  short overview of Kleene algebra, the formalism on which our semantic domains
%%  are based.

 \paragraph{Program model.}
 A \emph{program} consists of a finite set of procedures $\{P_i\}_{1 \leq i
   \leq N}$.  A \emph{procedure} is a pair $P_i = \langle G_i, LV_i \rangle$,
 where $G_i$ is a flow graph and $LV_i$ is a set of local variables.  A
 \emph{flow graph} is a finite directed graph $G_i = \langle V_i, E_i,
 \entry_i, \exit_i \rangle$ equipped with a distinguished \emph{entry vertex}
 $\entry_i \in V_i$ and \emph{exit vertex} $\exit_i \in V_i$.  We assume that
 every vertex in $V_i$ is reachable from $\entry_i$, $\entry_i$ has no
 incoming edges, $\exit_i$ has no outgoing edges, and that the vertices
 and local variables of each procedure are pairwise disjoint.

 To simplify our discussion, we associate with each edge $e \in E =
 \bigcup_{1 \leq i \leq N}E_i$ an \emph{action} $\act(e) \in \Act$ that gives a
 syntactic representation of the meaning of that edge, where $\Act$ is defined
 as
   \[\Act ::= x \texttt{ := } t \mid \texttt{assume($t$)} \mid \texttt{return} \mid \texttt{call $i$} \text{ where $1 \leq i \leq N$}\]
   \[s,t \in \Exp ::= x \mid n \text{ where $n \in \mathbb{Z}$} \mid s \bullet t \text{ where $\bullet \in \{+,-,*,/\}$}\]
   \[x \in \Var \supseteq \bigcup_{1 \leq i \leq N} LV_i\]
   We will sometimes use \texttt{[$t$]} in place of \texttt{assume($t$)} (as
   in Figure~\ref{fig:example}).  Note that our model has no parameters or
   return values in procedure calls, but these can be modeled using global
   variables.

 %% Each edge $e \in \bigcup_{1 \leq i \leq N}E_i$ is associated with an
 %% \emph{action} $\act(e) \in \Act$ that gives a syntactic representation of the
 %% meaning of that edge.  To simplify our discussion, we will fix a particular
 %% set of actions:

 %% We assume that if $e \in E_i$ (i.e., $e$ is an edge belonging to procedure
 %% $i$), no variable in $LV_j$ with $j \neq i$ is free in $\act(e)$.  Note that
 %% our model has no parameters or return values in procedure calls, but these
 %% can be modeled using global variables
  Let $G = \langle V, E, \entry, \exit \rangle$ be a flow graph.  For an edge
  $e \in E$, we use $\src(e)$ to denote the source of $e$ and $\tgt(e)$ to
  denote its target.  We define a path $\pi = e_1 \dotsi e_n \in E^*$ to be a
  finite sequence of edges such that for each $i < n$, $\tgt(e_i) =
  \src(e_{i+1})$.  We lift the $\src/\tgt$ notation from edges to paths by
  defining $\src(\pi) = \src(e_1)$ and $\tgt(\pi) = \tgt(e_n)$.  For any $u, v
  \in V$, we define $\paths{u}{v}$ to be the set of paths $\pi$ of $G$ such
  that $\src(\pi) = u$ and $\tgt(\pi) = v$. 

 \paragraph{Path expressions}
 We will assume familiarity with regular expressions.  We use the following
 syntax for the regular expressions $\regexp_\Sigma$ over some alphabet
 $\Sigma$: $p,q$ denote regular expressions, $\epsilon$ denotes the empty
 word, $\emptyset$ denotes the empty language, $p+q$ denotes choice, $p\cdot
 q$ denotes sequencing, and $p^*$ denotes iteration.

 %% \begin{align*}
 %%   p,q \in \regexp_\Sigma ::=&\; a \text{ for any } a \in \Sigma\\
 %%   |&\; \epsilon & \text{empty word}\\
 %%   |&\; \emptyset & \text{empty language}\\
 %%   |&\; p+q & \text{choice}\\
 %%   |&\; p \cdot q & \text{sequential composition}\\
 %%   |&\; p^* & \text{iteration}\\
 %% \end{align*}

 A \emph{path expression} for a flow graph $G = \langle V, E, \entry, \exit
 \rangle$ is a regular expression $p \in \regexp_E$ over the alphabet of edges
 of $G$, such that each word in the language generated by $p$ corresponds to a
 path in $G$ \cite{Tarjan1981,Tarjan1981b}.  For any pair of vertices $u,v$,
 there is a (not necessarily unique) path expression that generates
 $\paths{u}{v}$, the set of paths from $u$ to $v$.  A particular case of
 interest is are path expressions which, for a given vertex $v$, generates the
 set $\paths{\entry}{v}$ of all control flow paths from the entry vertex of
 $G$ to $v$.  Such a path expression can be seen as a succinct representation
 of the set of all computations of $G$ ending at $v$.

  The \emph{(single-source) path expression problem} for a flow graph $G =
  \langle V, E, \entry, \exit \rangle$ is to compute for each vertex $v \in V$
  a path expression $\pathexp{\entry}{v} \in \regexp_E$ representing the set
  of paths $\paths{\entry}{v}$.  An efficient algorithm for solving this
  problem was given in \cite{Tarjan1981}; a more recent path expression
  algorithm appears in \cite{Scholz2007}.  Understanding the specifics of
  these algorithms is not essential to our development, but the idea behind
  both is to divide $G$ into subgraphs, use Gaussian elimination to solve the
  path expression problem within each subgraph, and then combine the
  solutions.  Kleene's classical algorithm for converting a finite automaton
  to a regular expression \cite{Kleene1956} provides another way of solving
  path-expression problems that, while less efficient than
  \cite{Tarjan1981,Scholz2007}, should provide adequate intuition for readers
  not familiar with these algorithms.

 \section{Interpretations} \label{sec:interp}
 %% There are two particularly important notions of a concrete semantics in a
 %% second-order framework.  The first is the relational interpretation, in which
 %% program executions are identified with their input-output relations.  That
 %% is, the relational interpretation of a program takes the meaning of a program
 %% $P$ to be a set $\sem{P}^R \subseteq \Env \times \Env$, where $\Env = \Var
 %% \rightarrow \mathbb{Z}$ is the set of environments.  Intuitively, $\langle
 %% \rho, \rho' \rangle \in \sem{P}^R$ if and only if, starting with the initial
 %% environment $\rho$, there is an execution of $P$ that terminates with
 %% environment $\rho'$.  The important second interpretation is the (finite)
 %% trace interpretation, in which program executions are finite sequences of
 %% states, that represent (in addition to the input and output state) every
 %% intermediate state that is reached along the execution.

 In our framework, meaning is assigned to programs by an
 \emph{interpretation}.  We use interpretations to define both program analyses
 and the concrete semantics which justify their correctness.  Interpretations
 are composed of two parts: (i) a semantic \emph{domain} of interpretation, which
 defines a set of possible program meanings and which is equipped with
 operations for composing these meanings,
 and (ii) a \emph{semantic function} that associates a meaning to each atomic
 program action.  In this section, we define interpretations formally.

 %% and
 %% provide two concrete interpretations, the \emph{path interpretation} and
 %% the \emph{relational interpretation}, which form the semantic basis of many
 %% program analyses.

 We will limit our discussion in this section to the \emph{intraprocedural}
 variant of interpretations.  We will assume that a program consists of a
 single procedure and that the flow graph associated with that procedure has
 no \texttt{call} or \texttt{return} actions.  We will continue to make this
 assumption until Section~\ref{sec:interproc}, where we introduce our
 interprocedural framework.

 Our semantic domains are algebraic structures endowed with the familiar
 regular operators of sequencing, choice, and iteration.  These algebraic
 structures are used in place of the (complete) lattices typically used in the
 iterative theory of abstract interpretation.  We are interested in two
 particular variations: \emph{pre-Kleene algebras} (\PKA{}s) and
 \emph{quantales} \cite{Mulvey1986}.  Pre-Kleene algebras use weaker
 assumptions than the ones for Kleene algebras \cite{Kozen1991}, while
 quantales make stronger ones.  We make use of both variants for three
 reasons. (I) \PKA{}s are desirable because we wish to make our framework as
 broadly applicable as possible, so that it can be used to design and prove
 the correctness of a large class of program analyses.\footnote{In fact, the
   \PKA{} assumptions are weaker than those of any similar algebraic structure
   that have been proposed for program analysis that we are aware of: Kleene
   algebras used in \cite{Bouajjani2003}, the bounded idempotent semirings
   used in weighted pushdown systems \cite{Reps2005}, the \emph{flow algebras}
   of \cite{Filipiuk2011}, and the \emph{left-handed} Kleene algebras of
   \cite{Kot2004}.} (II) Quantales are desirable because we can prove a
 precision theorem (Theorem~\label{thm:intrap_mop}) comparable to the
 classical coincidence theorem of dataflow analysis \cite{Kam1977}. (III) Due
 to a result we will present in Section~\ref{sec:intra_abstract}, it is
 particularly easy to prove the correctness of a program analysis over a
 \PKA{} with respect to a concrete semantics over a quantale (as is the case,
 for example, in our motivating example of linear recurrence analysis).

%% that give analysis designers considerable
%%  flexibility in defining program analyses.  Quantales are a strengthening of
%%  the Kleene algebra axioms, which are considerably more restricted than
%%  \PKA{}s, but also have some desirable properties which make them convenient
%%  for describing concrete semantics.

 We begin by formally defining pre-Kleene algebras:
 \begin{definition}[\PKA{}]
   A \emph{pre-Kleene algebra} ($\PKA{}$) is a 6-tuple $\mathcal{K} = \langle
   K, \plus, \mul, \closure, 0, 1 \rangle$, where $\langle K, \plus, 0
   \rangle$ forms join-semilattice with least element 0 (i.e., $\plus$ is a
   binary operator that is associative, commutative, and idempotent, and which
   has 0 as its unit), $\langle K, \mul, 1 \rangle$ forms a monoid (i.e.,
   $\mul$ is an associative binary operator with unit 1), $\closure$ is a
   unary operator, and such that the following hold:\\
   \begin{minipage}[t]{0.5\textwidth}
   \begin{align*}
     (a \mul b) \plus (a \mul c) &\leq a \mul (b \plus c) & \mul \text{ is left pre-distributive}\\
     (a \mul c) \plus (b \mul c) &\leq (a \plus b) \mul c & \mul \text{ is right pre-distributive}\\
     1 \plus (a \mul a^\closure) &\leq a^\closure & (\emph{I1})\\
     1 \plus (a^\closure \mul a) &\leq a^\closure & (\emph{I2})
   \end{align*}
   \end{minipage}

   \noindent where $\leq$ is the order on the semilattice $\langle K, \plus, 0
   \rangle$.  The priority of the operators is the standard one ($\closure$
   binds tightest, followed by $\mul$ then $\plus$).
 \end{definition}

 As is standard in abstract interpretation, the order on the domain should be
 conceived as an approximation order: $a \leq b$ iff $a$ is approximated by
 $b$ (i.e., if $a$ represents a more precise property than $b$).  Note that
 the pre-distributivity axioms for $\mul$ are equivalent to $\mul$ being
 monotone in both arguments.

\begin{example} \label{ex:lra_domain}
   We return to our motivating example, linear recurrence analysis.  We define
   the semantic algebra of this analysis as:
   \[\mathcal{L} = \tuple{L, \lplus, \lmul, \lstar, 0^\logi, 1^\logi}\]
   The set of transition formulae $L$ and the $\lplus$, $\lmul$, and $\lstar$
   operators are defined as in Section~\ref{sec:motivation}.  The constants
   $0^\logi$ and $1^\logi$ are defined as\\
   \noindent\hspace*{1cm}$0^\logi = \formula{\false}$ \hfill $1^\logi = \formula{\bigwedge_{x \in \Var} x = x'}$\hspace*{1cm}

   It is easy to check that $\mathcal{L}$ forms a \PKA{}.  We will show only
   that (I1) holds as an example.  Let $a \in L$.  Recalling the definition of
   $\lstar$ from Section~\ref{sec:motivation}, we have
   \begin{equation*}
     \phi^\lstar = \formula{\exists k. k \geq 0 \land \bigwedge_{x \in
         I(\phi)} x' = cl_x}
   \end{equation*}
   where for each induction variable $x \in I(\phi)$, $cl_x$ is a closed form
   for an affine recurrence $x' = x + f(\vec{y})$ such that $\phi \models x' =
   x + f(\vec{y})$.

   It is easy to see that $1^\logi \leq \phi^\star$ (since $\leq$ is logical
   implication in this algebra and for all $x$, $cl_x[0/k] = x$).  It remains
   to show that $\phi \lmul \phi^\lstar \leq \phi^\lstar$.  Noting that
   $\lmul$ is monotone and (by the definition of $I$ and $cl_x$) $\phi \leq
   \formula{\bigwedge_{x \in I(\phi)} x' = x + f(\vec{y})}$, we have
   \begin{align*}
     a \lmul a^\lstar &\leq \formula{\bigwedge_{x \in I(\phi)} x' = x + f(\vec{y})} \lmul a^\lstar\\
     &=\formula{\parbox{0.625\linewidth}{$\exists_{x \in I(\phi)} x''. \big(\bigwedge_{x \in I(\phi)} x'' = x + f(\vec{y})\big)$ \\$\land \big(\exists k. k \geq 0 \land \bigwedge_{x \in I(\phi)} x' = cl_x[z''/z]\big)$}}\\
     &= \formula{\exists k. k \geq 0 \land \bigwedge_{x \in I(\phi)} x' = (cl_x[z + f(\vec{z})/z])}\\
     &= \formula{\exists k. k \geq 1 \land \bigwedge_{x \in I(\phi)} x' = cl_x}\\
     &\leq \phi^\lstar
   \end{align*}
\end{example}

 Quantales are a strengthening of the axioms of Kleene algebra that are
 particularly useful for presenting concrete semantics.    Intuitively, we may think of
 quantales as Kleene algebras where the choice and iteration operators are
 ``precise.''

 \begin{definition}[Quantale]
   A \emph{quantale} $\mathcal{K} = \langle K, \plus, \mul, \closure, 0, 1
   \rangle$ is a \PKA{} such that $\langle K, \plus, 0\rangle$ forms a
   complete lattice, and for any $a \in K$, $a^\closure = \bigoplus_{n \in
     \mathbb{N}} a^n$, and such that $\mul$ distributes over arbitrary sums:
   that is, for any index set $I$, we have
   \begin{align*}
     a\mul\big(\bigoplus_{i \in I(\phi)} b_i\big) &= \bigoplus_{i \in I} (a\mul
     b_i) \\
     \big(\bigoplus_{i \in I} a_i\big)\mul b &= \bigoplus_{i \in I} (a_i\mul b)
   \end{align*}
 \end{definition}

 \begin{example} \label{ex:rel_domain}
   For any set $A$, the set of binary relations on $A$ forms a quantale.  A
   case of particular interest is binary relations over the set $\Env = \Var
   \rightarrow \mathbb{Z}$ of program environments.  This quantale is the
   semantic algebra of the concrete semantics which justifies the correctness
   of linear recurrence analysis (i.e., it fills the typical role of state
   collecting semantics in the iterative abstract interpretation framework).
   Formally, this quantale is defined as:
   \[ \mathcal{R} = \tuple{\powerset{\Env \times \Env}, \relplus, \relmul, \relstar, 0^\relational, 1^\relational} \]
 where

   \hspace*{0pt}\hfill$R\relmul S = \{ \langle \rho, \rho'' \rangle : \exists \rho'. \langle \rho, \rho' \rangle \in R \land \langle \rho', \rho'' \rangle \in S \}$\hfill\hspace*{0pt}

 \begin{minipage}[t]{4cm}
   \begin{align*}
     0^\relational &= \emptyset\\
     1^\relational &= \{ \langle \rho, \rho \rangle : \rho \in \Env \}     
   \end{align*}
 \end{minipage}
 \begin{minipage}[t]{4cm}
   \begin{align*}
     R \relplus S &= R \cup S\\
     R^\relstar &= \bigcup_{n \in \mathbb{N}} R^n
   \end{align*}
 \end{minipage}
 \end{example}
 
 %% For many natural domains the stronger axioms of \KA{} \emph{will} hold (in
 %% particular, in the domains of paths and relations in this section), but
 %% requiring that all interpretations do is too stringent.  For an example of
 %% the practical limitations of such a requirement, consider that $1 + xx^* =
 %% x^*$ is a theorem of \KA{}, so loop unrolling cannot affect the precision of
 %% any analysis defined by interpretation over a \KA{}-domain.

 We note that $\mathcal{L}$, the semantic algebra of linear recurrence
 analysis defined in Example~\ref{ex:lra_domain}, is \emph{not} an example of
 a quantale.  In particular, it does not form a complete lattice and its
 iteration operator does not satisfy the condition that for all $a$, $a^\lstar
 = \bigoplus^\logi_{n \in \mathbb{N}} a^n$.  It is also interesting to note
 that $\mathcal{L}$ does not meet the assumptions of \emph{any} of the
 Kleene-algebra-like structures that have been previously proposed for use in
 program analysis \cite{Bouajjani2003,Reps2005,Filipiuk2011,Kot2004}.

%% However, multiplication does
%%  distribute over \emph{finite} sums (and infinite sums where the least upper
%%  bound is defined).

 Finally, we may formally define our notion of interpretation, which will be
 used to define both concrete semantics and program analyses in our framework.

 \begin{definition}[Interpretation] \label{def:interpretation}
   An \emph{interpretation} is a pair $\mathscr{I} = \langle \mathcal{D},
   \sem{\cdot} \rangle$, where $\mathcal{D}$ is a set equipped with the
   regular operations (e.g., a pre-Kleene algebra or a quantale) and
   $\sem{\cdot} : E \rightarrow \mathcal{D}$.  We will call $\mathcal{D}$ the
   \emph{domain} of the interpretation and $\sem{\cdot}$ the \emph{semantic
     function}.
 \end{definition}

 Interpretations assign meanings to programs by ``evaluating'' path
 expressions within the interpretation.  For any interpretation $\mathscr{I} =
 \langle \langle D, \plus, \mul, \closure, 0, 1 \rangle, \sem{\cdot} \rangle$
 and any path expression $p$, we use $\mathscr{I}[p]$ to denote the
 interpretation of $p$ within the domain $D$, which can be defined recursively
 as:\\
 \begin{minipage}{4cm}
 \begin{align*}
   \mathscr{I}[e] &= \sem{e}\\
   \mathscr{I}[\emptyset] &= 0\\
   \mathscr{I}[\epsilon] &= 1
 \end{align*}   
 \end{minipage}
 \begin{minipage}{4cm}
 \begin{align*}
   \mathscr{I}[p+q] &= \mathscr{I}[p] \plus \mathscr{I}[q]\\
   \mathscr{I}[p\cdot q] &= \mathscr{I}[p] \mul \mathscr{I}[q]\\
   \mathscr{I}[p^*] &= \mathscr{I}[p]^\star
 \end{align*}
 \end{minipage}

 A first example of an interpretation is the one defining linear
 recurrence analysis, obtained by pairing the \PKA{} defined in
 Example~\ref{ex:lra_domain} with the semantic function defined in
 Section~\ref{sec:motivation}.  We finish this section with a second example,
 the relational interpretation, which we will use in the next section to
 justify the correctness of linear recurrence analysis in Section~\ref{sec:intra_abstract}.
\begin{example} \label{ex:rel_interp}
  In the relational interpretation, the meaning of a program is its
  input/output relation.  That is, the meaning of a program is the set of all
  $\tuple{\rho, \rho'} \in \Env \times \Env$ such that there is an execution
  of the program that, starting from environment $\rho$, terminates in
  environment $\rho'$.  Formally, the relational interpretation $\mathscr{R}$
  is defined as:
 \[ \mathscr{R} = \tuple{\mathcal{R}, \sem{\cdot}^\relational} \]

 \noindent where $\mathcal{R}$ is as in Example~\ref{ex:rel_domain} the semantic function is defined by
 \begin{equation*}
   \sem{e}^\relational =
   \begin{cases}
     \{ \langle \rho, \rho[x \gets \mathcal{E}\sem{t}(\rho)] \rangle : \rho \in \Env \} & \text{if } \act(e) = x \texttt{ := } t\\
     \{ \langle \rho, \rho \rangle : \rho \in \Env, \rho \models t \neq 0\} & \text{if } \act(e) = \texttt{[}t\texttt{]}
   \end{cases}
 \end{equation*}
 (using $\mathcal{E}\sem{t}(\rho)$ to denote the integer to which the
 expression $t$ evaluates in environment $\rho$)
\end{example}

\subsection{Correctness and precision} \label{sec:intra_correct}

 We now justify the correctness of using path-expression algorithms for
 program analysis.  Given a control flow graph and a vertex $v$ of the graph,
 a path expression $\pathexp{\entry}{v}$ is a regular expression representing
 the set of paths from $\entry$ to $v$.  A reasonable assumption is that the
 interpretation of $\pathexp{\entry}{v}$ approximates the interpretation of
 each of these paths.  The following theorem states that this is the case, if
 the semantic algebra of the interpretation is a \PKA{}.
 \begin{theorem}[Correctness] \label{thm:intrap_mop}
   Let $\mathscr{I} = \langle \mathcal{D}, \sem{\cdot} \rangle$ be an
   interpretation where $\mathcal{D}$ is a \PKA{}, $G = \langle V, E, \entry,
   \exit \rangle$ be a flow graph, and $v \in V$ be a vertex of $G$.  Then for
   any path $e_1 \dotsi e_n \in \paths{\entry}{v}$, we have
   \[ \sem{e_1} \mul \dotsi \mul \sem{e_n} \leq \mathscr{I}[\pathexp{\entry}{v}] \]
 \end{theorem}
 \begin{proof}
   We prove that for any word $w = e_1 \dotsi e_n$ that is recognized by
   $\pathexp{\entry}{v}$ and any path expression $p$ that recognizes $w$, we
   have
   \[ \sem{w} \leq \mathscr{I}[p] \]
   (where $\sem{w}$ denotes $\sem{e_1} \mul \dotsi \mul \sem{e_n}$) by
   induction on $p$.  The main result then follows since the language
   recognized by $\pathexp{\entry}{v}$ is exactly $\paths{\entry}{v}$.

   \vspace{0.25cm}

   \noindent Case $p = \emptyset$: contradiction, since $w$ is recognized by $p$

   \vspace{0.1cm}
   \noindent Case $p = \epsilon$:  then $n=0$, and $\mathscr{I}[p] = 1 = \sem{w}$

   \vspace{0.1cm}
   \noindent Case $p = e$ for some edge $e$:  then $n=1$, $e_1 = e$, and
   $\mathscr{I}[p] = \sem{e} = \sem{w}$

   \vspace{0.1cm}
   \noindent Case $p = q + r$: wlog, we may assume that $w$ is accepted by
   $q$.  By the induction hypothesis, $\sem{w} \leq \mathscr{I}[q]$, whence
   \[\sem{w} \leq \mathscr{I}[q] \plus \mathscr{I}[r] = \mathscr{I}[q+r]\]

   \vspace{0.1cm}
   \noindent Case $p = q \cdot r$: then there exists $w_1,w_2$ such that
   $w=w_1w_2$, $w_1$ is recognized by $q$, and $w_2$ is recognized by $r$.  By
   the induction hypothesis, $\sem{w_1} \leq \mathscr{I}[q]$ and $\sem{w_2}
   \leq \mathscr{I}[r]$.  By monotonicity of $\mul$, \[\sem{w} = \sem{w_1} \mul
   \sem{w_2} \leq \mathscr{I}[q] \mul \mathscr{I}[r]\]

   \vspace{0.1cm}
   \noindent Case $p = q^\closure$: Since $w$ is recognized by $q^\closure$,
   there exists some $n$ such that $w = w_1 \dotsi w_n$ and each $w_i$ is
   recognized by $q$.  By the induction hypothesis, $\sem{w_i} \leq
   \mathscr{I}[q]$ for each $i$.  It follows that
   \[\sem{w} = \sem{w_1} \mul \dotsi \mul \sem{w_n} \leq
   \mathscr{I}[p]^n\] Finally, we have $\mathscr{I}[p]^n \leq
   \mathscr{I}[p]^\closure = \mathscr{I}[p^\closure]$ by
   Lemma~\ref{lem:star_induction}.
 \end{proof}

 The preceding theorem is analogous to the classical result which was used to
 justify the use of chaotic iteration algorithms for data flow analysis: the
 fixed-point solution to an analysis approximates the merge-over-paths
 solution \cite{Kildall1973}.  Our framework also accommodates a precision
 theorem, which is analogous to the coincidence theorem from dataflow analysis
 \cite{Kam1977}.  This result states that any path-expression algorithm
 computes exactly the least upper bound of the interpretation of the paths in
 the case that the semantic algebra of the interpretation is a quantale.
 \begin{theorem}[Precision] \label{thm:intrap_mop}
   Let $\mathscr{I} = \langle \mathcal{D}, \sem{\cdot} \rangle$ be an
   interpretation where $\mathcal{D}$ is a quantale, $G = \langle V, E,
   \entry, \exit \rangle$ be a flow graph, and $v \in V$ be a vertex of $G$.
   Then
   \[ \bigoplus_{\pi \in \paths{\entry}{v}} \sem{\pi} = \mathscr{I}[\pathexp{\entry}{v}] \]
     where for a path $\pi = e_1\dotsi e_n$, we define $\sem{\pi} = \sem{e_1}
     \mul \dotsi \mul \sem{e_n}$.
 \end{theorem}
Our precision theorem implies that for any path expressions $p$ and $p'$ that
recognize the same set of paths, and any interpretation $\mathscr{I}$ over a
quantale, we have $\mathscr{I}[p] = \mathscr{I}[p']$.  This is important from
the algorithmic perspective since it implies that the result of an analysis
over a quantale is the same regardless of how path expressions are computed
(e.g., we may use any of the algorithms appearing in
\cite{Kleene1956,Tarjan1981,Scholz2007}).  While this ``algorithm
irrelevance'' property would be undesirable in general,\footnote{Algorithm
  irrelevance would imply (among other things) that loop unrolling cannot
  increase analysis precision, since unrolling preserves the set of recognized
  paths.} it would be un-intuitive if the concrete semantics of a program
depended on such a detail.

%!TEX root = ./paper.tex 
\section{Abstraction} \label{sec:intra_abstract}

 We now develop a theory of abstract interpretation for (intraprocedural)
 compositional program analyses in our algebraic framework.  There are a
 vast number of ways to approach the problem of establishing an approximation
 relationship between one interpretation and another \cite{Cousot1992}.  This
 section presents an adaptation of one such method, soundness relations, to
 our framework.

%% We believe
%%  that developing an abstract interpretation framework that is specialized to
%%  the path-expression methodology is crucial to its widespread acceptance and
%%  contributes to a deeper understanding of program analyses based on path
%%  expressions.

 %% In the full version of this paper we also give a second method
 %% based on Galois connections \cite{fullversion}.
 
 Given two \PKA{}s $\mathcal{C}$ and $\mathcal{A}$, a soundness relation is a
 relation $\soundrel \subseteq \mathcal{C} \times \mathcal{A}$ that defines
 which concrete properties (elements of $\mathcal{C}$) are approximated by
 which abstract properties (elements of $\mathcal{A}$).  That is,
 $\sound{c}{a}$ indicates that $a \in \mathcal{A}$ is a sound approximation of
 $c \in \mathcal{C}$.  We require that soundness relationships are
 congruent with respect to the regular operations; for example, if $a \in
 \mathcal{A}$ approximates $c \in \mathcal{C}$, then
 $a^{\closure_{\mathcal{A}}}$ should approximate $c^{\closure_{\mathcal{C}}}$.
 Formally, we define soundness relations as follows:

 \begin{definition}[Soundness relation] \label{def:sr} Given two \PKA{}s
   $\mathcal{C}$ and $\mathcal{A}$, $\soundrel \subseteq \mathcal{C} \times
   \mathcal{A}$ is a \emph{soundness relation} if
   $\sound{0_{\mathcal{C}}}{0_{\mathcal{A}}}$,
   $\sound{1_{\mathcal{C}}}{1_{\mathcal{A}}}$, and for all $\sound{c}{a}$
   and $\sound{c'}{a'}$ we have the following:
   
   \noindent \hspace*{0.5cm}$\sound{c \plus_{\mathcal{C}} c'}{a \plus_{\mathcal{A}} a'}$ \hfill
    $\sound{c \mul_{\mathcal{C}} c'}{a \mul_{\mathcal{A}} a'}$ \hfill
    $\sound{c^{\closure_{\mathcal{C}}}}{a^{\closure_{\mathcal{A}}}}$ \hspace*{0.25cm}

   %% Additionally, $\soundrel$ must respect the approximation orders of
   %% $\mathcal{C}$ and $\mathcal{A}$, in the sense that
   %% \begin{itemize}
   %% \item If $\sound{c}{a}$ and $c' \leq_{\mathcal{C}} c$, then $\sound{c'}{a}$
   %% \item If $\sound{c}{a}$ and $a \leq_{\mathcal{A}} a'$, then $\sound{c}{a'}$.
   %% \end{itemize}
 \end{definition}

 %% \begin{definition}[Abstract interpretation] Given two interpretation
 %%   $\mathscr{I}^{\mathcal{C}}, \mathscr{I}^{\mathcal{A}}$ and a soundness relation $\soundrel
 %%   \subseteq \mathcal{D}^{\mathcal{C}} \times \mathcal{D}^{\mathcal{A}}$ we say that
 %%   $\mathscr{I}^{\mathcal{A}}$ is an \emph{abstract interpretation} of
 %%   $\mathscr{I}^{\mathcal{C}}$ if there exists a soundness relation such that for
 %%   all $e \in E$, $\sound{\natsem{e}}{\abssem{e}}$.
 %% \end{definition}

 Algebraically speaking, $\soundrel \subseteq \mathcal{C} \times \mathcal{A}$ is
 a soundness relation iff it is a subalgebra of the direct product \PKA{}
 $\mathcal{C} \times \mathcal{A}$.
 %% that is upwards closed under the
 %% \emph{consistency ordering} $\preceq$ defined by \[\langle c, a \rangle
 %% \preceq \langle c', a' \rangle \iff c \geq_{\mathcal{C}} c' \land a
 %% \leq_{\mathcal{A}} a'.\]

 The main result on soundness relations is the following abstraction
 theorem, which states that if we can exhibit a soundness relation $\soundrel$
 between a concrete interpretation $\mathscr{C}$ and an abstract
 interpretation $\mathscr{A}$, then the interpretation of any path expression
 in $\mathscr{C}$ is approximated by its interpretation in $\mathscr{A}$
 according to $\soundrel$.  From a program analysis perspective, we can think
 of the conclusion of this theorem as \emph{the program analysis defined by
   $\mathscr{A}$ is correct with respect to the concrete semantics defined by
   $\mathscr{C}$}.

 \begin{theorem}[Abstraction] \label{thm:sound_abstract}
   Let $\mathscr{C}$ and $\mathscr{A}$ be interpretations over domains
   $\mathcal{C}$ and $\mathcal{A}$, let $\soundrel \subseteq \mathcal{C}
   \times \mathcal{A}$ be a soundness relation, and let $p \in \regexp_E$ be a path
   expression.  If for all $e \in E$, $\sound{\natsem{e}}{\abssem{e}}$, then $\sound{\mathscr{C}[p]}{\mathscr{A}[p]}$. Moreover, this result holds even if $\mathcal{C}$ and
   $\mathcal{A}$ do not satisfy \emph{any} of the \PKA{} axioms.
 \end{theorem}
 \begin{proof}
   Straightforward, by induction on $p$.
 \end{proof}

 A common special case of the soundness relation framework is when the
 concrete domain is a quantale and the abstract domain is a \PKA{} (e.g., this
 is the case when proving the correctness of linear recurrence analysis
 w.r.t. the relational interpretation).  Under this assumption, the following
 proposition establishes a set of conditions that imply that a given relation
 is a soundness relation, but which may be easier to prove than using
 Definition~\ref{def:sr} directly.  In particular, this proposition relieves
 the analysis designer from having to prove congruence of the iteration
 operator.
 \begin{proposition} \label{prop:ka_sr}
   Let $\mathcal{C}$ be a quantale, $\mathcal{A}$ be a \PKA{}, and let
   $\soundrel \subseteq \mathcal{C} \times \mathcal{A}$ be a relation
   satisfying the following:
   \begin{enumerate}
   \item $\sound{1_\mathcal{C}}{1_{\mathcal{A}}}$
   \item If $\sound{c}{a}$ and $\sound{c'}{a'}$, then $\sound{c
     \mul_{\mathcal{C}} c'}{a \mul_{\mathcal{A}} a'}$
   \item If $\sound{c}{a}$ and $a \leq_{\mathcal{A}} a'$, then
     $\sound{c}{a'}$.
   \item For all $S \subseteq \mathcal{C}$, if for all $c \in S$,
     $\sound{c}{a}$, then $\sound{\bigoplus_{\mathcal{C}} S}{a}$.
   \end{enumerate}
   Then $\soundrel$ is a soundness relation.
 \end{proposition}
 \begin{proof}
   Congruence of $0$ and $\plus$ are straightforward, so we just prove
   congruence of iteration.  First, we prove a lemma:

   \begin{lemma}[Star induction] \label{lem:star_induction}
     Let $\mathcal{D}$ be a \PKA{} and let $p \in \mathcal{D}$.  Then for any
     $n \in \mathbb{N}$, $p^n \leq p^\closure$.
   \end{lemma}
   \begin{proof}
     By induction on $n$.  For the base case $n = 0$, we have \[p^0 = 1 \leq
     p^\closure\] by (I1).

     For the induction step, assume that $p^n \leq p^\closure$ and calculate
     \begin{align*}
       p^{n+1} &= p^n \mul p & \text{definition}\\
       &\leq p^\closure \mul p & \text{induction hypothesis, monotonicity}\\
       &\leq p^\closure & \text{(I2)}
     \end{align*}
   \end{proof}

   Now we return to proving congruence of iteration. Suppose $\soundrel$
   satisfies the conditions of Proposition~\ref{prop:ka_sr}, and that $\sound{c}{a}$.  It is easy to prove that
   $\sound{c^n}{a^n}$ for all $n$ by induction on $n$, using the first two
   conditions for $\soundrel$.  It follows from the
   Lemma~\ref{lem:star_induction} that $a^n \leq a^{\closure_{\mathcal{A}}}$
   for all $n$, so we have by the third condition
   $\sound{c^n}{a^{\closure_{\mathcal{A}}}}$ for all $n$.  Finally, the fourth
   condition implies
   $\sound{c^{\closure_{\mathcal{C}}}}{a^{\closure_{\mathcal{A}}}}$ (since
   $\mathcal{C}$ is a quantale, $c^{\closure_{\mathcal{C}}}$ is the least
   upper bound of the set $\{ c^n : n \in \mathbb{N} \}$).
 \end{proof}
 
 We conclude our discussion of soundness relations with an example of how they
 can be used in practice.
 
 \begin{example}
   The correctness of linear recurrence analysis is proved with respect to the
   relational semantics (Examples~\ref{ex:rel_domain} and~\ref{ex:rel_interp}), where the soundness
   relation $\soundrel \subseteq \mathcal{R} \times \mathcal{L}$ is given by
   \begin{displaymath}
     \sound{R}{\varphi} \iff \forall \langle \rho, \rho' \rangle \in R,
     [\rho,\rho'] \models \varphi
   \end{displaymath}
   where $[\rho,\rho']$ denotes the structure such that $[\rho,\rho'](x) =
   \rho(x)$ for every unprimed variable $x$, and $[\rho,\rho'](x') = \rho'(x)$
   for every primed variable $x'$.  We may show that $\soundrel$ is a soundness
   relation by Proposition~\ref{prop:ka_sr}.  Conditions \emph{1} and \emph{2}
   (congruence of $1$ and $\mul$) are trivial.  Conditions \emph{3} and
   \emph{4} follow easily from the definition of $\soundrel$.

   Finally, it is easy to see that for any edge $e$, we have
   $\sound{\sem{e}^\relational}{\lsem{e}}$, so by
   Theorem~\ref{thm:sound_abstract}, for any path expression $p \in \regexp_E$
   we have $\sound{\mathscr{R}[p]}{\mathscr{L}[p]}$.
 \end{example}

 \section{Interprocedural interpretations} \label{sec:interproc}
 A major problem in program analysis is designing analyses that handle
 procedures both efficiently and precisely.  One of the primary benefits of
 compositional program analyses is that they can be easily adapted to the
 interprocedural setting by employing procedure summarization.  The idea
 behind this approach is to compute a summary for each procedure $i$ which
 overapproximates its execution and then uses these summaries to interpret
 \texttt{call} $i$ actions \cite{Sharir1981}.

 %% This summary-based approach is inherently second-order (i.e., it can be
 %% justified with respect to the second-order semantics of a program, but not
 %% its first-order collecting semantics), which makes our second-order
 %% path-expression framework a natural choice to solve analyses based on
 %% computing summaries.

 One of the challenges involved in interprocedural analysis is the need to
 deal with local variables.  The difficulty lies in the fact that the values
 of local variables before a procedure call are restored after the procedure
 call.  Our approach uses an algebraic quantification operator, inspired by
 algebraic logic \cite{Henkin1971,Halmos1962}, to intuitively \emph{move
   variables out of scope}.  This is a binary operation $\exists x. a$ that
 takes a variable $x$ and a property $a$, and yields a property where the
 variable $x$ is ``out of scope'', but which otherwise approximates $a$.

 Interprocedural analysis presents another challenge that is specific to our
 algebraic approach: the regular operations (choice, sequencing, and
 iteration) are not sufficient for describing interprocedural behaviours,
 since the language of program paths with appropriately matched calls and
 returns (i.e., \emph{interprocedural paths}) is not regular in the presence
 of recursion.  Quantales already have sufficient structure to be appropriate
 semantic domains for interprocedural analyses, but \PKA{}s need to be
 augmented with a widening operator.

 We will now formally define interprocedural semantic domains and
 interpretations.  We will explain the motivation behind these definitions and
 how they solve the problems of local variables and recursion in the rest of
 this section.
 \begin{definition}[Quantified Pre-Kleene Algebra]
   A \emph{quantified pre-Kleene algebra} (\QPKA{}) over a set of variables
   $\Var$ is a tuple $\mathcal{K} = \langle K, \plus, \mul, \closure, \exists,
   \widen, 0, 1 \rangle$, such that $\langle K, \plus, \mul, \closure, 0, 1
   \rangle$ is a \PKA{}, and for the quantification operator $\exists : \Var \times K \rightarrow K$ and 
   the widening operator $\widen : K \times K \rightarrow K$ the following hold:
   \begin{align*}
     (\exists x.a) \plus (\exists x.b) &\leq \exists x. (a \plus b) & \text{\emph{(Q1)}}\\
     \exists x. ((\exists x. a) \mul b) &= (\exists x. a) \mul (\exists x. b) & \text{\emph{(Q2)}}\\
     \exists x. (a \mul (\exists x. b)) &= (\exists x. a) \mul (\exists x. b) & \text{\emph{(Q3)}}\\
     \exists x. \exists y. a &= \exists y. \exists x. a & \text{\emph{(Q4)}}\\
     a \plus b &\leq a \widen b  & \text{\emph{(W1)}}
   \end{align*}
   and for any sequence $\langle a_i \rangle_{i \in \mathbb{N}}$ in $K$, the sequence
   $\tuple{w_i}_{i \in \mathbb{N}}$ defined by
   \begin{align*}
     w_0 &= a_0\\
     w_n &= w_{n-1} \widen a_n
   \end{align*}
   eventually stabilizes (there exists some $k \in \mathbb{N}$ such that for
   all $j > k$, $w_k = w_j$).
 \end{definition}

 We discuss the intuition behind the quantification operator in
 Section~\ref{sec:locals}.  The conditions on quantification come from the
 treatment of existential quantifiers in algebraic logic
 \cite{Halmos1962,Henkin1971} (hence our use of ``quantification'' for our
 scoping operator).  The conditions for the widening operator are standard
 \cite{Cousot1977}.  We adapt quantales to the interprocedural setting in a
 similar fashion, except that we omit the widening operator and require that
 scoping distribute over arbitrary sums.

 \begin{definition}[Quantified Quantale]
   A \emph{quantified quantale} over a set of variables $\Var$ is a tuple
   $\mathcal{K} = \tuple{K, \plus, \mul, \closure, \exists, 0, 1}$ such that
   $\tuple{K,\plus,\mul,\closure,0,1}$ is a quantale and $\exists : \Var
   \times K \rightarrow K$ is an operator such that \emph{(Q2)-(Q4)} hold and
   for any $x \in \Var$ and index set $I$,
   \begin{align*}
     (\exists x. \bigoplus_{i \in I} a_i) &= \bigoplus_{i \in I} (\exists x. a_i) & \text{\emph{(Q1')}}\\
   \end{align*}
 \end{definition}

 Since the order of quantification does not matter (due to (Q4)), for any
 finite set of variables $X = \{x_1, \ldots, x_n\}$ we will write $\exists
 X.a$ to denote $\exists x_1.\exists x_2.\dotsi \exists x_n. a$.

 Recall that we use $E$ to denote the set of control flow edges of all
 procedures, $E_i$ to denote the edges of procedure $i$, and $LV_i$ to denote
 the local variables of procedure $I$.  Interpretations can be adapted to
 interprocedural semantic domains as follows:

 %% The quantification operations presented here are weaker than
 %% the ones for cylindric and polyadic algebra \cite{Henkin1971,Halmos1962}. In
 %% fact, some\footnote{The only cylindric algebra quantification axiom that
 %%   generally does \emph{not} hold is (C2): $a \leq \exists x. a$} of these
 %% stronger axioms do often hold (in particular, for the quantification operator
 %% in the relational interpretation presented earlier in this section).
 %% However, since (Q1-Q4) are the only axioms necessary to prove the
 %% results in this paper, we omit the stronger axioms.
%% We also assume that the semantic function of an interprocedural
%%  interpretation respects a free variable restriction.
 \begin{definition}[Interprocedural interpretation]
   An \emph{interprocedural interpretation} is a pair $\mathscr{I} = \langle
   \mathcal{D}, \sem{\cdot} \rangle$, where $\mathcal{D}$ is an
   interprocedural semantic domain (a quantified \PKA{} or quantale),
   $\sem{\cdot} : E \rightarrow \mathcal{D}$, and such that for any $e \in
   E_i$ and variable $x \in LV_j$ (with $j \neq i$), $\exists x. \sem{e} =
   \sem{e}$ (i.e., the actions belonging to procedure $i$ do not read from or
   write to the local variables of procedure $j$).
 \end{definition}

 \subsection{Local variables} \label{sec:locals}
 %% We provide intuition on our way of handling local variables by exhibiting the
 %% quantification operator in the relational interpretation.

 We now define the quantification operator for the relational interpretation
 to provide intuition on how our framework deals with local variables.
 Any relation $R \subseteq \Env \times \Env$ partitions the set of variables
 into a \emph{footprint} and a \emph{frame}.  The variables in the footprint
 are the ones that are relevant to the relation (in the sense that the
 variable was read from or written to along some execution represented by the
 relation); the frame is the set of variables that are irrelevant.  Variable quantification intuitively moves a variable from the
 footprint to the frame.  Formally, we define this operator as:
   \[ \exists^\relational x. R = \{ \langle \rho[x \gets n], \rho'[x \gets n] \rangle : \langle \rho, \rho' \rangle \in R, n \in \mathbb{Z}\} \]

 %% \begin{figure}
 %%   \texttt{foo()}\\ \texttt{x := g}\\ \texttt{if (g < 10)}\\
 %%   \hspace*{0.5cm}\texttt{g := g + x}\\
 %%   \hspace*{0.5cm}\texttt{call foo}\\
 %%   \hspace*{0.5cm}\texttt{g := g - x}\\   
 %%   \caption{A recursive procedure demonstrating the problem with local
 %%     variables}
 %%   \label{fig:local_vars}
 %% \end{figure}
   \newpage
\begin{wrapfigure}{r}{2.5cm}
  \footnotesize
  \fbox{\begin{minipage}{2.25cm}
   \texttt{foo():}\\
   \hspace*{0.25cm}\texttt{x := g}\\
   \hspace*{0.25cm}\texttt{if (g < 10)}\\
   \hspace*{0.5cm}\texttt{g := g + x}\\
   \hspace*{0.5cm}\texttt{call foo}\\
   \hspace*{0.5cm}\texttt{g := g - x}    
  \end{minipage}}
 \end{wrapfigure}
   For example, consider the recursive procedure \texttt{foo} shown to the
   right.  The variable \texttt{g} is global and \texttt{x} is local to the
   procedure \texttt{foo}.  The input/output relation of this function is:
   \[ R = \{ \langle \rho, \rho' \rangle : \rho'(\texttt{x}) = \rho(\texttt{g}) = \rho'(\texttt{g}) > 10)\}
   \]

   We wish to verify the correctness of this input/output relation by checking that if $R$ is used as
   the interpretation of the recursive call, the interpretation of the path
   expression from the entry to the exit of \texttt{foo} is again $R$.  This
   check fails because the local variable \texttt{x} is conflated with the
   copy of \texttt{x} from the recursive call (i.e., we incorrectly determine
   that the value of \texttt{x} changes across the recursive call).  The correct interpretation of the call is obtained if \texttt{x} is moved out of scope by using
   \[ \exists^{\mathscr{R}} x. R = \{ \langle \rho, \rho \rangle : \rho(\texttt{g}) > 10 \} \]
   as the interpretation of the call edge.

\begin{wrapfigure}{r}{2.5cm}
  \footnotesize
  \fbox{\begin{minipage}{2.25cm}
   \texttt{bar():}\\
   \hspace*{0.25cm}\texttt{if (x == 0)}\\
   \hspace*{0.5cm}\texttt{assert(flag)}\\
   \hspace*{0.25cm}\texttt{else}\\
   \hspace*{0.5cm}\texttt{x := 0}\\
   \hspace*{0.5cm}\texttt{call bar}\\
   \hspace*{0.25cm}\texttt{flag := 0}\\
  \end{minipage}}
 \end{wrapfigure}
   A more subtle issue with local variables is illustrated in
   procedure \texttt{bar} shown to the right.  Here, \texttt{flag} is a global
   variable and \texttt{x} is local.  Suppose that we wish to compute an
   input/output relation representing the executions of \texttt{bar} which end
   at the assertion.  The set of \emph{interprocedural paths} ending at the
   assertion can be described by the following regular expression:
   \[ (\texttt{assume(x != 0); x := 0; call bar})^*\texttt{; assume(x == 0)} \]

   The relational interpretation of this expression\footnote{using $\{
     \tuple{\rho, \rho[\texttt{flag} \gets 0]} : \rho \in \Env \}$ as the
     interpretation of \texttt{call bar}} is:
   \[ \{ \tuple{\rho,\rho} : \rho(\texttt{x}) \geq 0 \} \]
 
   This relation (incorrectly) indicates that the value of $\texttt{flag}$ in
   the pre-state will be the same as the value of $\texttt{flag}$ in the
   post-state (so if \texttt{bar} is called from state where $\texttt{flag}$
   is nonzero, the assertion will not fail).  The problem is that the
   \texttt{assume(x == 0)} instruction refers to a different instance of the
   local variable \texttt{x} than the one in \texttt{x := 0} in the regular
   expression above.  By applying the quantification operator to move
   \texttt{x} out of scope in the loop, we arrive at the correct
   interpretation:
   \[ \{ \tuple{\rho,\rho} : \rho(\texttt{x}) \geq 0 \} \cup \{ \tuple{\rho,\rho[\texttt{flag} \gets 0]} : \rho(\texttt{x}) \geq 0 \} \]

   A significant benefit of our treatment of local variables using
   quantification operators is that we can handle both these potential issues
   with a single operator, which reduces the burden of designing and
   implementing an interprocedural analysis.

%% \footnote{If $R$ is used to
%%      interpret the \texttt{call foo} action, the resulting input/output
%%      relation is\\$\{\tuple{\rho,\rho'} : \rho'(x) = \rho'(g) = \rho(g) \geq
%%      10 \lor \rho(g) < 10 \land \rho'(x) = 2\rho(g) > 0 = \rho'(g)\}$},

 \subsection{Recursion}

 As mentioned before, the behaviours of recursive programs may be non-regular languages, and beyond the expressive power of choice, sequencing, and iteration operators. 
 Here we provide two solutions to this problem, which emerge from methods for
 defining a solution to a system of recursive equations in the classical
 theory of abstract interpretation \cite{Cousot1977}.  For (quantified) quantales, the
 underlying order forms a complete lattice, so we can use Tarski's fixpoint
 theorem to guarantee the existence of a least fixpoint solution to the
 equation system.  For quantified \PKA{}s, we use a widening
 operator to accelerate convergence of the Kleene iteration sequence
 associated with the equation system.

 First, we factor the interpretation of call edges out of the interpretation
 of path expressions.  Let $\mathscr{I} = \langle \mathcal{D}, \sem{\cdot}
 \rangle$ be an (interprocedural) interpretation.  A \emph{summary assignment}
 (for a program consisting of $N$ procedures) is mapping $S : [1,N]
 \rightarrow \mathcal{D}$ that assigns each procedure a property in
 $\mathcal{D}$ (i.e., the procedure summary).  Given any path expression $p$,
 we use $\iinterp{\mathscr{I}}{S}{p}$ to denote the interpretation of $p$
 within $\mathscr{I}$, using $S$ to interpret calls (this function is defined
 as $\mathscr{I}[p]$ was in Section~\ref{sec:interp}, except edges which
 correspond to \texttt{call} instructions are interpreted using $S$ as the
 semantic function).

%% This function can be defined
%%  recursively as:

%%  \begin{eqnarray*}
%%    \iinterp{\mathscr{I}}{S}{e} &=& 
%%    \begin{cases}
%%      S(i) & \text{if } \act(e) = \texttt{call } i\\
%%      \sem{e} & \text{otherwise}
%%    \end{cases}\\
%%    \iinterp{\mathscr{I}}{S}{\emptyset} &=& 0\\
%%    \iinterp{\mathscr{I}}{S}{\epsilon} &=& 1\\
%%    \iinterp{\mathscr{I}}{S}{p+q} &=& \iinterp{\mathscr{I}}{S}{p} \plus \iinterp{\mathscr{I}}{S}{q}\\
%%    \iinterp{\mathscr{I}}{S}{p\cdot q} &=& \iinterp{\mathscr{I}}{S}{p} \mul \iinterp{\mathscr{I}}{S}{q}\\
%%    \iinterp{\mathscr{I}}{S}{p^*} &=& \iinterp{\mathscr{I}}{S}{p}^\closure
%%  \end{eqnarray*}

 A summary assignment $S$ is \emph{inductive} if for all procedures $i \in
 [1,N]$, we have (recalling that $LV_i$ denotes the set of local variables of
 procedure $i$)
 \[ S(i) \geq \exists LV_i. \iinterp{\mathscr{I}}{S}{\pathexp{\entry_i}{\exit_i}} \]
 \noindent Conceptually, inductive summary assignments are the ``viable
 candidates'' for approximating the behaviour of each procedure.  An equivalent
  way of describing inductive summary assignments is that they are the post-fixed points
%% \footnote{A post-fixed
%%    point of a function $f$ is a point $x$ such that $f(x) \leq x$}
of the
 function \[F : ([1,N] \rightarrow \mathcal{D}) \rightarrow ([1,N] \rightarrow \mathcal{D}),\] defined as:
 %% \begin{eqnarray*}
 %%   F &:& ([1,N] \rightarrow \mathcal{D}) \rightarrow ([1,N] \rightarrow \mathcal{D}) \\
 %%   F(S) &=& \lambda i. \exists LV_i
 %%   . \iinterp{\mathscr{I}}{S}{\pathexp{\entry_i}{\exit_i}}
 %% \end{eqnarray*}
 \[F(S) = \lambda i. \exists LV_i
   . \iinterp{\mathscr{I}}{S}{\pathexp{\entry_i}{\exit_i}}\]
%% The proof is by induction on the path expression
%%  $\pathexp{\entry_i}{\exit_i}$, using the fact that each regular operation is
%%  monotone in all its arguments, and for any variable $x$, $\exists x$ is
%%  monotone.

 In order to define the meaning of a program with procedures, we must be able
 to select a \emph{particular} post-fixed point of $F$.  We will use
 $\overline{S}$ to denote this uniquely-selected post-fixed point (regardless
 of the method used to select it).  Our two approaches for selecting
 $\overline{S}$ are as follows:

 \noindent \textbf{\emph{(Quantified) Quantales}}: If the semantic domain
 $\mathcal{D}$ of an interpretation is a quantale, then the underlying order
 $\tuple{\mathcal{D}, \leq}$ is a complete lattice.  It follows that the
 function space $[1,N] \rightarrow \mathcal{D}$ is also a complete lattice
 (under the point-wise ordering).  It is easy to show that (under the
 assumption that $\mathcal{D}$ is a \QPKA{} quantale), the function $F$
 defined above is monotone.  Since $F$ is a monotone function on a complete
 lattice, Tarski's fixed point theorem guarantees that $F$ has a least fixed
 point \cite{Tarski1955}.  We define $\overline{S}$ to be this least fixed
 point.

 \noindent \textbf{\emph{\QPKA{}s}}: If the semantic domain $\mathcal{D}$ of
 an interpretation is a \QPKA{}, then we have access to a widening operator
 $\widen$.  Consider the sequence $\tuple{S_i}_{i \in \mathbb{N}}$ defined by:
 \begin{align*}
   S_1(i) &= 0\\
   S_n(i) &= S_{n-1}(i) \widen (\exists
 LV_i. \iinterp{\mathscr{I}}{S_{n-1}}{\pathexp{\entry_i}{\exit_i}})
 \end{align*}
 Using the properties of the widening operator, we can prove that there exists
 some $k \in \mathbb{N}$ such that $S_{k} = S_{k+1}$, and that $S_k$ is an
 inductive summary assignment.  We define $\overline{S} = S_k$.  In practice,
 an efficient chaotic iteration strategy \cite{Bourdoncle1993} may be employed
 to compute $S_k = \overline{S}$.

 \begin{example}
   To extend linear recurrence analysis
   to handle procedures, we must define a quantification operator $\exists^{\mathscr{L}}$ and a
   widening operator $\widen^{\mathscr{L}}$.  The quantification operator mirrors the one from the
   relational interpretation:
   \[ \exists^{\mathscr{L}} x. \varphi = \formula{(\exists x \exists x'. \varphi) \land x = x'} \]

   \noindent As with the sequencing and choice operators, this quantification
   operator is precise (i.e., it loses no information from the perspective of
   the relational interpretation).
%% Here, the $\exists$ on the right hand side
%%    of the equality is syntax -- it denotes a syntactic construction on a
%%    logical formula, and does not necessarily require a costly quantifier
%%    elimination step.
%% Due to
%%    the close correspondence between $\exists^{\mathscr{R}}$ and
%%    $\exists^{\mathscr{L}}$, it is easy to prove that $\sound{R}{\varphi}$
%%    implies $\sound{\exists^{\mathscr{R}} x. R}{\exists^{\mathscr{L}}
%%      x. \varphi}$.

   There are several possible choices for the widening operator $\varphi_1
   \widen^{\mathscr{L}} \varphi_2$ on $\mathscr{L}$.  One natural choice is to
   leverage existing widening operators.  For example, we may compute the best
   abstraction of $\varphi_1$ and $\varphi_2$ as convex polyhedra, and then
   apply a widening operator from that domain \cite{Cousot1978}.
 \end{example}

 \section{Interprocedural analysis} \label{sec:inter_analysis}
 We now adapt the path expression algorithm to the interprocedural framework.
 Let $\mathscr{I} = \tuple{\mathcal{D}, \sem{\cdot}}$ be an interpretation
 such that $\mathcal{D}$ is equipped with a widening operator and let
 $\mathcal{P} = \{ P_i \}_{1 \leq i \leq N}$ be a program, where $P_1$ denotes
 the main procedure of $\mathcal{P}$.  The \emph{interprocedural path
   expression problem} is to compute for each procedure $P_k$ and each vertex
 $v$ of $P_k$ an approximation of the interprocedural paths
 from $\entry_1$ to $v$.  We use $\mathscr{I}\tuple{v}$ to denote this
 approximation.

%% , compute for each
%%  procedure $k \in [1,N]$ and each vertex $v$ in the flow graph for $k$ an
%%  \emph{interprocedural path expression} $\ipathexp{\entry_1}{v}$ that
%%  represents the set of interprocedural paths that begin at $\entry_1$ and end
%%  at $v$.  In the following discussion, we make this problem statement precise
%%  by defining interprocedural paths and path expressions, and then give an
%%  algorithm for solving the interprocedural path expression problem.

 An \emph{interprocedural path} $\pi = e_1 \dotsi e_n$ is a sequence of edges
 such that for all $i < n$, either $\tgt(e_i) = \src(e_{i+1})$, or $\tgt(e_i)$
 has an outgoing edge which is labeled \texttt{call $k$} and $\src(e_{i+1}) =
 \entry_k$ (for some $k$).  This definition of interprocedural paths differs
 from the standard one \cite{Sharir1981}: rather than requiring that every
 return edge be matched with a call edge, we allow each edge labeled
 \texttt{call $k$} to stand for the entire execution of procedure $P_k$.  The
 benefit of this approach is that the set of interprocedural paths from one
 vertex to another is a regular language (whereas for \cite{Sharir1981} it is
 context-free, but not necessarily regular).  Note that, since call edges are
 interpreted as complete executions of a procedure, transferring control from
 one procedure to another is not marked by a call edge.

 Any interprocedural path from $\entry_1$ to a vertex $v$ belonging to
 procedure $P_k$ can be decomposed into an interprocedural path from $\entry_1$
 to $\entry_k$ and an intraprocedural path from $\entry_k$ to $v$.  So we may
 compute $\mathscr{I}\tuple{v}$ as
 \[ \mathscr{I}\tuple{v} = \mathscr{I}\tuple{\entry_k} \mul \iinterp{\mathscr{I}}{\overline{S}}{\pathexp{\entry_k}{v}} \]
 As a result of this decomposition, we need only to solve the interprocedural
 path expression problem for procedure entry vertices.  This problem can be
 solved by applying Tarjan's path-expression algorithm to the \emph{call
   graph} of $\mathcal{P}$.

 The \emph{call graph} $CG_{\mathcal{P}} = \langle \mathcal{P}, CE \rangle$
 for a program $\mathcal{P}$ is a directed graph whose vertices are the
 procedures $\{P_i\}_{1 \leq i \leq N}$, and such that there is an edge
 (called a \emph{call edge}) $P_i \rightarrow P_j \in CE$ iff there exists
 some edge $e \in E_i$ such that $\act(e) = \texttt{call } j$.  To apply
 Tarjan's algorithm to $CG_{\mathcal{P}}$, we must define an interpretation
 for call edges.  The interpretation of a call edge $P_i \rightarrow P_j \in
 CE$ should be an approximation of all paths that begin at the entry of $P_i$
 and end in a call of $P_j$.  The local variables of $P_i$ go out of scope
 when the program enters $P_j$, so variables in $LV_i$ should be removed from
 the footprint of $P_i \rightarrow P_j$.  Formally,
 \begin{equation*}
 \begin{split}
 \iinterp{\mathscr{I}}{\overline{S}}{P_i \rightarrow P_j} =  \exists LV_i. \bigoplus \{ \iinterp{\mathscr{I}}{\overline{S}}{\pathexp{\entry_i}{\src(e)}} :  e \in E_i,\\ \act(e) = \texttt{call } j \}  
 \end{split}
 \end{equation*}

 We formalize the interprocedural program analysis algorithm we have developed
 in this section in Algorithm~\ref{alg:interproc}.

 \begin{algorithm}
   \caption{Interprocedural path-expressions} \label{alg:interproc}
   \begin{algorithmic}
     \Require An interpretation $\mathscr{I}$ and a program $\mathcal{P} = \{P_i\}_{1 \leq i \leq N}$
     \Ensure An array $\texttt{PathTo}$ s.t. $\forall k \in [1,N]$ and each
      vertex $v$ of $k$, $\texttt{PathTo}[k][v] = \mathscr{I}\tuple{v}$.
     \State $S \gets \lambda i. 0$
     \Repeat
     \Comment{Compute procedure summaries}
     \State $S' \gets S$
     \State $S \gets \lambda i. S(i) \widen \exists LV_i. \iinterp{\mathscr{I}}{S}{\pathexp{\entry_i}{\exit_i}}$
     \Until{$S = S'$}
     \ForAll{$i \in [1,N]$}
     \Comment Compute paths to procedure entry points
     \State $\texttt{PathTo}[i][\entry_i] \gets \iinterp{\mathscr{I}}{S}{\pathexp{P_1}{P_i}}$
     \EndFor
     \ForAll{$i \in [1,N]$}
     \Comment{Compute paths to each vertex}
     \ForAll{$v \in V_i \setminus \{\entry_i\}$}
     \State $\texttt{PathTo}[i][v] \gets \texttt{PathTo}[i][\entry_i] \mul \iinterp{\mathscr{I}}{S}{\pathexp{\entry_i}{v}}$
     \EndFor
     \EndFor
     \State \Return $\texttt{PathTo}$
   \end{algorithmic}
 \end{algorithm}

 \paragraph{Example: Interprocedural Linear Recurrence Analysis}

 We now present an example of how our interprocedural analysis algorithm works
 using linear recurrence analysis on a simple program.  Consider the program
 pictured in Figure~\ref{fig:interproc_ex}, which includes two procedures:
 \texttt{main} (the entry point of the program) and \texttt{foo}.  The action
 of this program is to set \texttt{g} to 20 and then decrement it 10 times.
 The decrement loop is implemented with a recursive procedure \texttt{foo}
 which uses its parameter to keep track of the loop count (note that, since in
 our model procedures do not have parameters, we model parameter passing via
 assignment to the global variable \texttt{p0}).  We will illustrate how
 linear recurrence analysis can be used to prove that the assertion in
 \texttt{foo} (that \texttt{g} is positive) will always succeed.

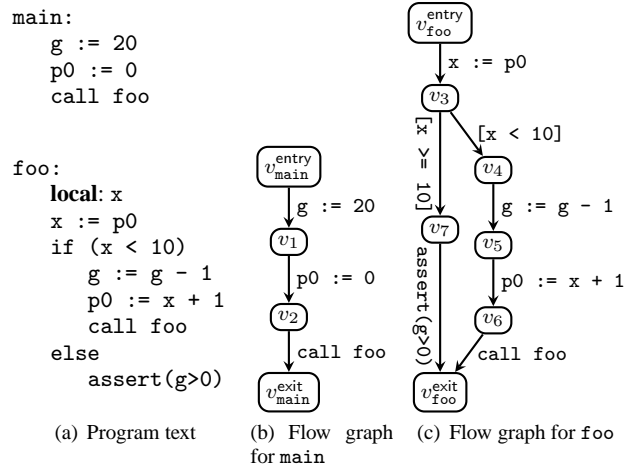
\begin{figure}
  \begin{center}
  \subfigure[Program text]{
    \begin{minipage}[b]{3cm}
\texttt{main:}\\
\hspace*{0.5cm}\texttt{g := 20}\\
\hspace*{0.5cm}\texttt{p0 := 0}\\
\hspace*{0.5cm}\texttt{call foo}\\

\vspace*{0.25cm}
\texttt{foo:}\\
\hspace*{0.5cm}\textbf{local}: \texttt{x}\\
\hspace*{0.5cm}\texttt{x := p0}\\
\hspace*{0.5cm}\texttt{if (x < 10)}\\
\hspace*{1cm}\texttt{g := g - 1}\\
\hspace*{1cm}\texttt{p0 := x + 1}\\
\hspace*{1cm}\texttt{call foo}
\hspace*{0.5cm}\texttt{else}\\
\hspace*{1cm}\texttt{assert(g>0)}\\
    \end{minipage}
  }
  \subfigure[Flow graph for \texttt{main}]{
  \begin{tikzpicture}[base,node distance=1cm]
  \footnotesize
  \node [rectangle,rounded corners,draw] (entry) {$\entry_{\texttt{main}}$};
  \node [rectangle,rounded corners,draw,below of=entry] (v1) {$v_1$};
  \node [rectangle,rounded corners,draw,below of=v1] (v2) {$v_2$};
  \node [rectangle,rounded corners,draw,below of=v2] (v3) {$\exit_{\texttt{main}}$};

  \path (entry) edge[->] node[right]{\texttt{g := 20}} (v1);
  \path (v1) edge[->] node[right]{\texttt{p0 := 0}} (v2);
  \path (v2) edge[->] node[right]{\texttt{call foo}} (v3);
  \end{tikzpicture}}
  \subfigure[Flow graph for \texttt{foo}]{
  \begin{tikzpicture}[base,node distance=1cm]
  \footnotesize
  \node [rectangle,rounded corners,draw] (entry) {$\entry_{\texttt{foo}}$};
  \node [rectangle,rounded corners,draw,below of=entry] (v3) {$v_3$};
  \node [rectangle,rounded corners,draw,below right of=v3,yshift=-0.25cm] (v4) {$v_4$};
  \node [rectangle,rounded corners,draw,below of=v4] (v5) {$v_5$};
  \node [rectangle,rounded corners,draw,below of=v5] (v6) {$v_6$};
  \node [rectangle,rounded corners,draw,below left of=v6,yshift=-0.25cm] (exit) {$\exit_{\texttt{foo}}$};

  \node [rectangle,rounded corners,draw,below of=v3,node distance=1.75cm] (v7) {$v_7$};

  \path (entry) edge[->] node[right]{\texttt{x := p0}} (v3);
  \path (v3) edge[->] node[right]{\texttt{[x < 10]}} (v4);
  \path (v4) edge[->] node[right]{\texttt{g := g - 1}} (v5);
  \path (v5) edge[->] node[right]{\texttt{p0 := x + 1}} (v6);
  \path (v6) edge[->] node[right]{\texttt{call foo}} (exit);
  \path (v3) edge[->,sloped] node[below]{\texttt{[x >= 10]}} (v7);
  \path (v7) edge[->,sloped] node[below]{\texttt{assert(g>0)}} (exit);
  \end{tikzpicture}}

  \end{center}
  \caption{An example program with procedures for interprocedural linear recurrence analysis.} \label{fig:interproc_ex}
\end{figure}

We begin by computing an inductive summary assignment $\overline{S}$.  For the
sake of this example, let us assume that we use an extremely imprecise
widening operator, defined as
\[ \phi \widen \phi' = 
\begin{cases}
  \phi & \text{if } \phi \iff \phi'\\
  true & \text{otherwise}
\end{cases}\]

First, we compute (using a path-expression algorithm
\cite{Tarjan1981,Scholz2007}), a path expression for each procedure, which
represents the set of paths from the entry of that procedure to its exit.  One
possible result is:
\[
  \pathexp{\entry_{\texttt{main}}}{\exit_{\texttt{main}}} = \tuple{\entry_{\texttt{main}},v_1} \tuple{v_1,v_2} \tuple{v_2,\exit_{\texttt{main}}}
\]
\begin{align*}
 \pathexp{\entry_{\texttt{foo}}}{\exit_{\texttt{foo}}} =  \tuple{\entry_{\texttt{foo}},v_3} \Big(& (\tuple{v_3,v_4} \tuple{v_4,v_5} \tuple{v_5,v_6} \tuple{v_6,\exit_{\texttt{foo}}})\\
& + \tuple{v_3,v_7}\tuple{v_7,\exit_{\texttt{foo}}}\Big)
\end{align*}

We compute a summary assignment for each procedure by repeatedly interpreting
these path expressions in $\mathscr{L}$, using the summary for $\texttt{foo}$
to interpret calls to \texttt{foo}, until the summary assignment converges.
This computation proceeds as follows:
 \begin{align*}
   S_0(\texttt{foo}) &= 0^\logi = \formula{\false}\\
   S_0(\texttt{main}) &= 0^\logi = \formula{\false}\\
   S_1(\texttt{foo}) &= 0^\logi \widen \big(\exists x. \iinterp{\mathscr{L}}{S_0}{\pathexp{\entry_{\texttt{foo}}}{\exit_{\texttt{foo}}}}] \big)\\
     &= 0 \widen \Big(\exists x.\formula{x' \geq 10 \land x' = p0 \land p0' = p0 \land g' = g}\Big)\\
     &= 0^\logi \widen \formula{p0 \geq 10 \land p0' = p0 \land g' = g \land x'=x}\\
     &= \formula{\true}\\
   S_1(\texttt{main}) &= 0^\logi \widen \iinterp{\mathscr{L}}{S_0}{\pathexp{\entry_{\texttt{main}}}{\exit_{\texttt{main}}}}]
     = 0^\logi\\
   S_2(\texttt{main}) &= 0^\logi \widen \iinterp{\mathscr{L}}{S_1}{\pathexp{\entry_{\texttt{main}}}{\exit_{\texttt{main}}}}]
     = \formula{\true}
 \end{align*}

\begin{wrapfigure}{r}{2cm}
  \footnotesize
  \hspace*{-0.75cm}
  \begin{tikzpicture}[node distance=1.5cm,>=stealth]
    \node [ellipse,draw] (main) {\texttt{main}};
    \node [ellipse,draw,right of=main] (foo) {\texttt{foo}};
    \path (main) edge[->] (foo);
    \path [draw,->] (foo.270) .. controls +(270:0.5cm) and +(330:0.5cm) .. (foo.330);
  \end{tikzpicture}
 \end{wrapfigure}
At $S_2$ the fixedpoint iteration converges, and we define our inductive
summary assignment as $\overline{S} = S_2$.  Next, we compute summaries to the
entry point of each procedure, using the path expression algorithm on the call
graph shown to the right.  The interpretation of the edges is given by the
following:
 \begin{align*}
   \iinterp{\mathscr{L}}{\overline{S}}{\texttt{main} \rightarrow \texttt{foo}} &= \formula{g' = 20 \land p0' = 0 \land x' = x}\\
   \iinterp{\mathscr{L}}{\overline{S}}{\texttt{foo} \rightarrow \texttt{foo}} &= \fbox{\parbox{5cm}{$p0 < 0 \land p0' = p0 + 1 \land g' = g - 1$\\$\land x' = x$}}
 \end{align*}
which yields the following path-to-entry summaries:
 \begin{align*}
   \mathscr{I}\tuple{\entry_{\texttt{main}}} &= 1^\logi\\
   \mathscr{I}\tuple{\entry_{\texttt{foo}}} &= \mathscr{I}(\overline{S})[(\texttt{main} \rightarrow \texttt{foo})\cdot(\texttt{foo} \rightarrow \texttt{foo})^*]\\
   &=\formula{g' = 20 - p0' \land p0' \leq 10 \land x'=x}
 \end{align*}

Finally, we may compute a summary which overapproximates the executions to $v_7$:
 \begin{align*}
   \mathscr{I}\tuple{v_7} &= \mathscr{I}\tuple{\entry_{\texttt{foo}}} \lmul \iinterp{\mathscr{I}}{\overline{S}}{\pathexp{\entry_{\texttt{foo}}}{v_7}}\\
   &=\formula{g' = 20 - p0' \land p0' = 10 \land x' = p0'}
 \end{align*}
from which we may prove that the assertion in \texttt{foo} may never fail.

This example illustrates an interesting feature of our framework:
interprocedural loops (which result from recursive procedures such as
\texttt{foo}) are analyzed in exactly the same way as intraprocedural loops
(by simply applying an iteration operator).  This is a significant advantage of
our generic algebraic framework over defining a compositional analysis by
structural induction on the program syntax (as in
\cite{Ammarguellat1990,Popeea2006,Biallas2012}), which requires separate
treatment for interprocedural loops.

 \subsection{Correctness and precision} \label{sec:inter_correctness}

 It is easy to prove that Algorithm~\ref{alg:interproc} computes an array
 $\texttt{PathTo}$ such that for each procedure $P_k$ and each vertex $v$ of
 $P_k$, $\texttt{PathTo}[k][v] = \mathscr{I}\tuple{v}$.  As in
 Section~\ref{sec:intra_correct}, we justify the correctness of our
 interprocedural path-expression algorithm with respect to interprocedural
 paths.  We use $\mathsf{IVP}(v)$ to denote the set of interprocedurally valid
 paths (i.e., with matched calls and returns, as in \cite{Sharir1981}) to the
 vertex $v$, and for any $\pi \in \mathsf{IVP}(v)$, we use $\mathscr{I}[\pi]$
 to denote the interpretation of $\pi$ in $\mathscr{I}$.  The formal
 definition of $\mathscr{I}[\pi]$, as well as all proofs in this section, can
 be found in the extended version of this paper \cite{fullversion}.
 
 First, we have that Algorithm~\ref{alg:interproc} soundly approximates
 interprocedurally valid paths if the semantic algebra is a \QPKA{}:
 \begin{theorem}[Correctness]
   Let $\mathscr{I}$ be an interpretation over a \QPKA{}, $P_k$ be a procedure,
   $v$ be a vertex of $P_k$, and let $\pi \in \mathsf{IVP}(v)$ be an
   interprocedural path to $v$.  Then $\mathscr{I}[\pi] \leq
   \mathscr{I}\tuple{v}$.
 \end{theorem}

 Second, we have that Algorithm~\ref{alg:interproc} computes exactly the sum
 over interprocedurally valid paths in the case that the semantic algebra is a
 quantified quantale.
 \begin{theorem}[Precision]
   Let $\mathscr{I}$ be an interpretation over a (quantified) quantale, 
   $P_k$ be a procedure, $v$ be a vertex of $P_k$.  Then
   $\mathscr{I}\tuple{v} = \bigoplus_{\pi \in \mathsf{IVP}(v)} \mathscr{I}[\pi]$.
 \end{theorem}
 Note that the precision theorem above also addresses the subtle issues
 involving local variables, as in \cite{Knoop1992b,Lal2005} (e.g., it shows
 that we do not lose correlations between local and global variables due to
 procedure calls, as \cite{Lal2005} points out that \cite{Reps1995} does).

 \subsection{Abstraction} \label{sec:inter_abstract}
 In this section, we adapt the abstract interpretation framework of
 Section~\ref{sec:intra_abstract} to the interprocedural setting.  The
 conditions of Proposition~\ref{prop:ka_sr} are much easier to adapt to the
 interprocedural setting than the more general definition of a soundness
 relation, so we present this simplified version here.  Under the assumptions
 of Proposition~\ref{prop:ka_sr}, all that is required to extend soundness
 relations to the interprocedural setting is an additional congruence
 condition for the quantification operator (which ensures that ``abstract''
 quantification approximates ``concrete'' quantification).
 
 \begin{definition}[Interprocedural soundness relation]
   Let $\mathcal{C}$ be a quantified quantale and let $\mathcal{A}$ be a
   \QPKA{}.  A relation $\soundrel \subseteq \mathcal{C} \times \mathcal{A}$
   is a \emph{soundness relation} if $\soundrel$ satisfies the conditions of
   Proposition~\ref{prop:ka_sr}, and for all $\sound{c}{a}$, and all $x \in
   \Var$, we have $\sound{\exists_{\mathcal{C}} x. c}{\exists_{\mathcal{A}}
     x. a}$.
 \end{definition}

 The abstraction theorem carries over to the interprocedural setting,
 and is stated below.  The proof of this theorem is similar to the
 intraprocedural variant, except that we require some of the fixpoint
 machinery of classical abstract interpretation \cite{Cousot1977} due to the
 inductive summary assignment computation.

 \begin{theorem}[Abstraction] \label{thm:sound_abstraction}
   Let $\mathscr{C}, \mathscr{A}$ be interpretations, where the 
   domain of $\mathscr{C}$ is a quantified quantale and the domain of
   $\mathscr{A}$ is a \QPKA{}.  Let $\soundrel \subseteq \mathcal{C} \times
   \mathcal{A}$ be a soundness relation, and let $v$ be a vertex of some
   procedure.  If for all $e \in E$, $\sound{\natsem{e}}{\abssem{e}}$, then
   $\sound{\mathscr{C}\langle v \rangle}{\mathscr{A}\langle v \rangle}$.
 \end{theorem}

\section{The algebraic framework in practice}\label{sec:experiments}

In this section, we
report on the algorithmic side of our framework, using the linear recurrence
analysis presented in Section~\ref{sec:motivation} as a case study.  We
present experimental results which demonstrate two key points:
\begin{itemize}
\item Algorithm~\ref{alg:interproc}, the algorithmic basis of our
  interprocedural framework, can effectively (and efficiently) solve real
  program analysis problems in practice.
\item Compositional loop analyses (like linear recurrence analysis) can
  effectively compute accurate abstractions of loops.  This indicates that
  compositional loop analyses, a relatively under-explored family of analyses,
  are an interesting direction for future work by the program analysis
  community.  Our framework is particularly well suited for helping to explore
  this research program.
\end{itemize}

%% The compositional algebraic framework provides a clearly defined interface between the definition of a program analysis and the path-expression algorithms used to compute the result of the analysis, which allows for generic implementations of the path expression algorithms to be developed independent of the program analysis. To study the practical implications of our framework, we implemented one such path expression algorithm. 

%% As stated before, we believe that the algebraic framework presented in this paper presents the opportunity to develop more precise and/or more efficient program analyses for certain problem domains. Here, we present examples of two case studies as proof of concept to support this belief. The first one is an implementation of the {\em linear recurrence analysis} sketched in Section \ref{sec:motivation}, which illustrates how the compositional analyses of loops can result in more precise loop invariants. The second one is a dependence analysis, which illustrates the efficiency of an interprocedural analysis performed in the compositional framework over the iterative counterpart. Below, we briefly discuss the result of our experiments with these two analyses. 

\paragraph{Implementation}
One of the stated goals of our framework is to provide a clearly defined
interface between the definition of a program analysis and the algorithm used
to compute the result of the analysis.  The practical value of this interface
is that allows for generic implementations of path expression algorithms to be
developed independent of the program analysis.  This is particularly
interesting because the choice of path expression algorithm affects both the
speed and precision of an analysis.

Our interprocedural path expression algorithm (Algorithm~\ref{alg:interproc})
is implemented in OCaml as a functor parameterized by a module representing a
\QPKA{}.  This allows any analysis which can be designed in our framework to
be implemented as a plugin.  Linear recurrence analysis is implemented as one
such plugin.  For our implementation of linear recurrence analysis, we use Z3
\cite{z3} to resolve SMT queries that result from applying the iteration
operator and checking assertion violations.

\paragraph{Linear Recurrence Analysis}
Let us refer to our implementation of the linear recurrence analysis as {\sc
  LRA}. We compare the effectiveness of {\sc LRA} in proving properties of
program with loops with two state-of-the-art invariant generation tools: {\sc
  UFO} \cite{ufo} and {\sc InvGen} \cite{invgen}. Table \ref{tab:lra} presents
the results of this comparisons.  We ran the three tools on a suite consisting
of 73 benchmarks from the {\sc InvGen} test suite and 17 benchmarks from the
software verification competition (in particular, the safe, integer-only
benchmarks from the Loops category).  These benchmarks are small, but
difficult.  
\begin{table}[htdp]
\begin{center}
\small
\begin{tabular}{|l|c|c|c|} \hline

 Result           &  {\sc LRA}    & {\sc UFO}     & {\sc InvGen}                \\ \hline \hline

Safe           &  84 &  61 &     75\\
Unsafe         &   1 &   1 &      0\\
False positive &   5 &  18 &      9\\
False negative &   0 &   0 &      1\\
Timeout        &   1 &   5 &      0\\
Crash          &   0 &   6 &      3\\
\hline
\end{tabular}

\caption{Experimental Results comparing LRA analysis with {\sc UFO} and {\sc
    InvGen} on a set of benchmarks from {\sc InvGen} and SVComp.  {\sc InvGen}
  does not currently handle procedures, so benchmarks involving procedures are
  omitted from the {\sc InvGen} column.}
\end{center}
\label{tab:lra}
\end{table}%
The total results show that {\sc LRA} manages to prove more instances correct compared to {\sc UFO} and {\sc InvGen}.

 \section{Conclusion and Related work} \label{sec:rel_work}
 In this paper, we presented an algebraic framework for compositional program
 analysis.  Our framework can be used as the basis for analyses that use
 non-iterative methods of computing loop invariants and for efficient
 interprocedural analyses, even in the presence of recursion and local
 variables.  We close with a comparison of our framework with previous
 approaches to compositional analysis, uses of algebra in program analysis,
 and interprocedural analysis.

 \subsection{Compositional program analyses}
 Perhaps the most prominent use of compositionality is in interprocedural
 analyses based procedure summarization.  Summary-based analyses are not
 necessarily implementable in our framework since there are iterative methods
 for computing summaries that are beyond the scope of our framework.  The
 literature on this subject is too vast to list here, but we note one
 particular example which \emph{can} be implemented and proved correct in our
 framework: the affine equalities analysis of \cite{Muller-Olm2004}.

 Compositional loop analyses use a summary of the body of the loop in order to
 compute loop invariants.  The most common technique for presenting these
 analyses is by structural induction on program syntax
 \cite{Ammarguellat1990,Popeea2006,Biallas2012}.  An alternative technique is
 given in \cite{Kroening2008}, which is based on graph rewriting.  This
 technique is more generally applicable than induction on program syntax
 (e.g., it may be applied to programs with \texttt{goto}s), but is
 considerably more complicated: an analysis designer must develop a procedure
 which abstracts arbitrary loop-free programs (rather than a handful of simple
 syntactic constructors).  Our framework allows all of these analyses to be
 implemented in a way that is both simple and generally applicable.

 Elimination-style dataflow analyses are another source for compositional
 program analyses \cite{Ryder1986}.  Elimination-style dataflow analyses were
 designed to speed up iterative dataflow analyses by computing loop summaries
 and thereby avoid repeatedly propagating dataflow values through loops.  This
 style of analysis, the context in which Tarjan's path expression algorithm
 was originally developed \cite{Tarjan1981}.  The analyses considered in this
 line of work typically limited to the class of dataflow analyses useful in
 compiler optimization (e.g., gen/kill analyses).  These analyses can be
 implemented and proved correct in our framework, but we are primarily
 interested in a more general class of analysis (e.g., analyses for numerical
 invariant generation).

 \subsection{Algebra in program analysis}
 A number of papers have used algebraic structures as the basis of program
 analysis \cite{Reps2005,Kot2004,Bouajjani2003,Filipiuk2011}.  A summary of
 the relative strength of the assumptions on these structures is  presented in
 Figure~\ref{fig:strength}.  The notion of Pre-Kleene algebras introduced in
 this paper generalize all these structures, and thus provides a unifying
 foundation.

 Weighted pushdown-systems (WPDSs) are a generic tool for implementing
 interprocedural program analyses \cite{Reps2005}.  Weighted pushdown systems
 extend pushdown systems with a weight on each rule that is drawn from a
 bounded idempotent semiring (a Kleene algebra without an iteration operator
 that must satisfy the ascending chain condition).  Tarjan's path-expression
 algorithm has also found a use in the WPDS method in improving the
 $pre^*/post^*$ algorithms that drive WPDS-based analyses \cite{Lal2006}.

 There are two advantages of our framework over WPDSs. First, we compute loop
 summaries using an iteration operator, whereas WPDSs use fixed-point
 iteration.  The advantage of our approach using an iteration operator was
 also noted in \cite{Lal2006}.  The second advantage is conceptual simplicity:
 our framework is based on familiar regular expression operations and
 procedure summaries, whereas WPDSs require more sophisticated automata
 theory.

 There are two features of WPDSs that are \emph{not} currently handled by our
 approach: backwards analysis and stack-qualified queries.  Given the
 similarity between the WPDS framework and the interprocedural path-expression
 algorithm, we believe that our methodology could be adapted to handle these
 features as well.

 Bouajjani et al present a generic methodology for proving that two program
 locations are not coreachable in a concurrent program with procedures
 \cite{Bouajjani2003}.  Their method is based on developing a Kleene
 algebra\footnote{In fact, they are interested in two stronger forms: finite
   and commutative Kleene algebras.  Finite \KA{}s are quantales, but
   commutative \KA{}s are incomparable.} in which elements represent regular
 sets of paths; this allows the coreachability test to be reduces to
 emptiness-of-intersection of two regular languages.  Their method allows for
 significant flexibility in designing the abstract domain, but ultimately
 their work does not intend to be a completely generic framework, but rather
 to solve a specific problem in concurrent program analysis. 

 Kot \& Kozen address develop an algorithm for implementing second-order
 (i.e., trace-based) program analyses based on left-handed Kleene algebra
 \cite{Kot2004} (Kleene-algebras where sequencing is left distributive and
 right pre-distributive).  Their work is an implementation of a
 path-expression algorithm, which provides an alternative to the ones in
 \cite{Tarjan1981,Scholz2007} (and which, interestingly, uses a matrix
 construction on left-handed Kleene algebras rather than explicitly using
 graphs).  Their primary concern is with the implementation of analyses,
 rather than semantics-based justification of their correctness.

 Filipiuk et al present a program analysis framework most similar in spirit to
 ours, in which program analyses are defined by \emph{flow algebras}:
 pre-distributive, bounded, idempotent semirings (i.e., Pre-Kleene algebras
 without an iteration operator) \cite{Filipiuk2011}.  Of the works cited
 above, \cite{Filipiuk2011} is the only one which addresses the problem of
 proving the correctness of a program analysis with respect to a concrete
 semantics.  Since flow algebras omit an iteration operator, they cannot
 express the non-iterative loop abstractions that \PKA{}s can.

 \paragraph{Algebra in program semantics.}
 There is a line of work in program semantics which aims to enrich Kleene
 algebras with greater structure to make them useful as a basis for reasoning
 about programs: for example, \cite{Kozen2003,Hoare2009}. In contrast, we use
 a weaker structure than Kleene algebras (at least in terms of their
 axiomatization), but we typically assume that the operations of interest are
 computationally effective so that we may \emph{compute} the meaning of a
 program.  Work on concurrent Kleene algebra (with an additional parallel
 composition operator) \cite{Hoare2009} suggests an interesting direction for
 future research, e.g. adapting our framework to handle concurrency.

 \begin{figure}
   \begin{tikzpicture}
     \footnotesize
     \node (q) {Quantale};
     \node [below of=q] (ka) {Kleene algebra \cite{Bouajjani2003}};
     \node [below left of=ka,xshift=-1.25cm,yshift=-0.5cm] (lka) {Left-handed Kleene algebra \cite{Kot2004}};
     \node [below right of=ka,xshift=1.25cm] (bis) {Bounded idempotent semiring \cite{Reps2005}};
     \node [below of=bis] (fa) {Flow algebra \cite{Filipiuk2011}};
     \node [below of=ka,yshift=-1.5cm] (pka) {Pre-Kleene algebra};

     \path [draw] (q) -- (ka);
     \path [draw] (ka) -- (lka);
     \path [draw] (ka) -- (bis);
     \path [draw] (bis) -- (fa);
     \path [draw] (fa) -- (pka);
     \path [draw] (lka) -- (pka);
   \end{tikzpicture}   
   \caption{Relative strength of program analysis algebras}
   \label{fig:strength}
 \end{figure}
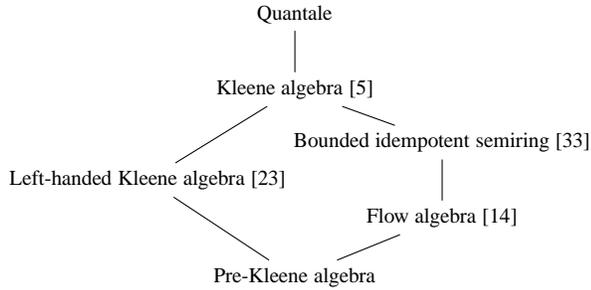

 \subsection{Interprocedural analysis}

 %% Our interprocedural analysis is a variant of the functional approach of
 %% Sharir \& Pnueli \cite{Sharir1981}.  Sharir \& Pnueli are primarily
 %% concerned with first-order analyses, and assume that the domain for procedure
 %% summaries is a function space over some first-order domain.  Our approach is
 %% more flexible, and only assumes the existence of a few algebraic operations
 %% on the domain.

%% , even in the case when the primary object of interest is
%%  first-order.  For example, the framework of Sharir and Pnueli cannot justify
%%  the correctness of analyses where exact function composition is not
%%  computable, but some \emph{approximate} composition operation is.

 The are a great number of approaches to interprocedural analysis; a nice
 categorization of these techniques can be found in \cite{Cousot2002}.  Two
 particular techniques of interest is the summary-based approach of Sharir \&
 Pnueli \cite{Sharir1981} and the method of Cousot \& Cousot for analyzing
 (recursive) procedures \cite{Cousot1977b}.  Both of these techniques assume
 that the domain for procedure summaries is a function space over some
 first-order domain.  Our approach  makes no such assumptions on the
 structure of the universe.

 \paragraph{Local variables.}
 Knoop \& Steffan extend the coincidence theorem of \cite{Sharir1981} to
 handle local variables by using a merge function to combine information about
 local variables from the calling context with information about global
 variables from the returning context \cite{Knoop1992b}.  Lal et al. use a
 similar approach to handle local variables in WPDS-based analyses \cite{Lal2005}. Cousot \&
 Cousot use a different approach to handling local variables by explicitly
 managing the footprint and frame of procedure summaries \cite{Cousot1977b};
 this is similar to the approach we use.  In contrast to \cite{Cousot1977b},
 our treatment of quantification is algebraic, whereas in \cite{Cousot1977b}
 moving a variable from the footprint to the frame is a particular concrete
 operation.

\bibliographystyle{abbrvnat}
\bibliography{paper}

%%% Local Variables: 
%%% mode: latex
%%% TeX-master: "paper"
%%% End:

\appendix

 \section{Proof of Theorem~\ref{thm:sound_abstraction}}
 \begin{proof}
   We begin with two general lemmas concerning interprocedural
   interpretations.  The main result is then a simple application of these two
   lemmas.
   \begin{lemma} \label{lem:pathexp_rel}
     Let $\mathscr{I} = \tuple{\mathcal{I}, \sem{\cdot}_{\mathcal{I}}}$ and
     $\mathscr{J} = \tuple{\mathcal{J},\sem{\cdot}_{\mathcal{J}}}$ be \QPKA{}
     interpretations, let $S_{\mathcal{I}}$ and $S_{\mathcal{J}}$ be summary
     assignments, and let $\soundrel \subseteq \mathcal{I} \times \mathcal{J}$ be
     any relation that is coherent with respect to the \QPKA{} operators and
     such that for all $e \in E$, 
     \[ \sound{\mathscr{I}(S_{\mathcal{I}})[e]}{\mathscr{J}(S_\mathcal{J})[e]} \]
     Then for any path expression $p$,
     \[ \sound{\mathscr{I}(S_{\mathcal{I}})[p]}{\mathscr{J}(S_\mathcal{J})[p]} \]

     Moreover, this holds even if $\mathcal{I}$ and $\mathcal{J}$ do not
     satisfy any of the \QPKA{} axioms.
   \end{lemma}
   \begin{proof}
     By induction on $p$.
   \end{proof}
   \begin{lemma} \label{lem:analysis_rel}
     Let $\mathscr{I} = \tuple{\mathcal{I}, \sem{\cdot}_{\mathcal{I}}}$ and
     $\mathscr{J} = \tuple{\mathcal{J},\sem{\cdot}_{\mathcal{J}}}$ be \QPKA{}
     interpretations, and let $\soundrel \subseteq \mathcal{I}
     \times \mathcal{J}$ be any relation that is coherent with respect to the
     \QPKA{} operators and such that for all $e \in E$,
     \[ \sound{\mathscr{I}(\overline{S}_{\mathcal{I}})[e]}{\mathscr{J}(\overline{S}_\mathcal{J})[e]} \]
     where $\overline{S}_{\mathcal{I}}$ and $\overline{S}_{\mathcal{J}}$ are
     inductive summary assignments (as defined in
     Section~\ref{sec:interproc}). Then for any vertex $v$,
     \[ \sound{\mathscr{I}\tuple{v}}{\mathscr{J}\tuple{v}} \]

     Moreover, this holds even if $\mathcal{I}$ and $\mathcal{J}$ do not
     satisfy any of the \QPKA{} axioms.
   \end{lemma}
   \begin{proof}
     This is essentially a corollary of Lemma~\ref{lem:pathexp_rel}.  Recall
     that
     \[ \mathscr{I}\tuple{v} = \mathscr{I}(\overline{S}_{\mathcal{I}})[\pathexp{P_1}{P_k}] \mul_\mathcal{I} \mathscr{I}(\overline{S}_{\mathcal{I}})[\pathexp{\entry_k}{v}] \]
     \[ \mathscr{J}\tuple{v} = \mathscr{J}(\overline{S}_{\mathcal{J}})[\pathexp{P_1}{P_k}] \mul_\mathcal{J} \mathscr{J}(\overline{S}_{\mathcal{J}})[\pathexp{\entry_k}{v}] \]
   So we need only to prove that
   \[ \sound{\mathscr{I}(\overline{S}_{\mathcal{I}})[\pathexp{P_1}{P_k}]}{\mathscr{J}(\overline{S}_{\mathcal{J}})[\pathexp{P_1}{P_k}]} \]
   and 
   \[ \sound{\mathscr{I}(\overline{S}_{\mathcal{I}})[\pathexp{\entry_k}{v}]}{\mathscr{J}(\overline{S}_{\mathcal{J}})[\pathexp{\entry_k}{v}]} \]
   The latter follows directly from Lemma~\ref{lem:pathexp_rel}.  It remains
   to show the former.

   Recall that the interpretation of a call graph edge $P_i \rightarrow P_j$
   is defined in terms of the interpretations of the path expression
   $\pathexp{\entry_i}{c}$ describing paths to a vertex $c$ that calls $j$.
   It follows from Lemma~\ref{lem:pathexp_rel} that we must have
   \[ \sound{\mathscr{I}(\overline{S}_{\mathcal{I}})[P_i \rightarrow P_j]}{\mathscr{J}(\overline{S}_{\mathcal{J}})[P_i \rightarrow P_j]} \]
   for all edges $P_i \rightarrow P_j$ in the call graph.  We may then prove by
   induction on $\pathexp{P_1}{P_k}$ that
   \[ \sound{\mathscr{I}(\overline{S}_{\mathcal{I}})[\pathexp{P_1}{P_k}]}{\mathscr{J}(\overline{S}_{\mathcal{J}})[\pathexp{P_1}{P_k}]} \]
   \end{proof}
   We can prove that for all $i$,
   $\sound{\overline{S}_{\mathcal{C}}(i)}{\overline{S}_\mathcal{A}(i)}$ in the
   standard way, by induction on the transfinite iteration sequence defining
   $\overline{S}_{\mathcal{C}}$ (using Lemma~\ref{lem:pathexp_rel} for the
   induction step).  We then have the main result by applying
   Lemma~\ref{lem:analysis_rel}.
 \end{proof}

\section{Coincidence}

 We now consider the question of when the algorithm presented in
 Section~\ref{sec:inter_analysis} computes the ideal (merge-over-paths)
 solution to an analysis problem.  This is analogous to the development of
 \emph{coincidence theorems} \cite{Kam1977,Sharir1981,Knoop1992b,Lal2005} in
 the field of dataflow analysis, which state conditions under which the
 iterative algorithm for solving dataflow analyses computes the ideal
 solutions (i.e., when the join-over-paths solution to a dataflow analysis
 problem coincides with the minimum fixpoint solution\footnote{In dataflow
   analysis literature, coincidence theorems typically relate the
   meet-over-paths and the maximum fixpoint solutions.  We follow in the
   tradition of abstract interpretation, which uses the dual order.}).  In this
 section, we formulate and prove a coincidence theorem for our interprocedural
 framework.

 The first step towards formulating our coincidence theorem is to define the
 join-over-paths solution to an analysis.  This requires us to define what a
 path is, and how to interpret it.  Our treatment of interprocedural paths in
 Section~\ref{sec:inter_analysis} is not adequate for this purpose, because it
 relies on interpreting call edges using procedure summaries (which represent
 a set of paths, rather than a single path).  In this section, we will take
 interprocedural paths to be words over an alphabet of control flow graph
 edges such that each return is properly matched with a call.  For a formal
 definition of such a path, the reader may consult \cite{Sharir1981}.  We use
 \textsf{IPaths(v)} to denote the set of all interprocedural paths to $v$.

 Now we must define the interpretation of an interprocedural path.  Let
 $\mathscr{I} = \langle \mathcal{D}, \sem{\cdot} \rangle$ be an
 interpretation.  We define the interpretation of a path within $\mathscr{I}$
 as a function that operates on \emph{stacks} of abstract values (as in
 \cite{Knoop1992b,Lal2005}).  We use the notation $[\langle V_1, a_1 \rangle,
   \ldots, \langle V_n, a_n \rangle]$ to denote a stack of $n$
 \emph{activation records}, where an activation record is a tuple $\langle
 V_i, a_i \rangle$ consisting of a set of variables $V_i$ that are local to
 that activation record, and an abstract value $a_i$.  The first entry in this
 list corresponds to the top of the stack.  We define our semantic function on
 edges so that each call edge corresponds to a push on this stack, each return
 edge corresponds to a pop, and each intraprocedural edge modifies the top of
 the stack (i.e., the ``current'' activation record).  Formally, we have

  \noindent$\widehat{\sem{e}}([\langle V_1, a_1 \rangle, \ldots, \langle V_n, a_n \rangle]) =$

   \hfill$\begin{cases}
     [\langle LV_i, 1 \rangle,\langle V_1, a_1 \rangle,\ldots,\langle V_n, a_n \rangle] & \text{if } \act(e) = \texttt{call } i\\
     [\langle V_2, a_2 \mul (\exists V_1.a_1) \rangle, \ldots, \langle V_n, a_n \rangle] & \text{if } \act(e) = \texttt{return}\\
     [\langle V_1, a_1 \mul \sem{e} \rangle, \ldots, \langle V_n, a_n \rangle] & \text{otherwise}
   \end{cases}$

   We may then associate with every path $\pi$ a stack of activation records
   $S(\pi)$ in the following way:
   \begin{eqnarray*}
     \textit{stack}(\epsilon) &=& [\langle LV_1, 1 \rangle]\\
     \textit{stack}(\pi a) &=& \widehat{\sem{a}}(\textit{stack}(\pi))
   \end{eqnarray*}

   Next, we define a $\textit{flatten}$ function that lowers a stack of
   activation records into a single abstract value, and existentially
   quantifies variables that are not in the scope of the current activation
   record.
   \[\textit{flatten}([\langle V_1, a_1 \rangle, \ldots, \langle V_n,
     a_n \rangle]) = (\exists V_n. a_n)\mul \dotsi \mul (\exists V_2. a_2)
   \mul a_1\]

   Finally, we take the interpretation of an interprocedural path within
   $\mathscr{I}$ to be $\textit{flatten}(\textit{stack}(\pi))$.

   Recalling the definition of $\mathscr{I}\langle v \rangle$ from
   Section~\ref{sec:inter_analysis}, we may state our coincidence theorem:
   \begin{theorem}[Coincidence] \label{thm:coincidence}
     Let $\mathscr{I} = \langle \mathcal{Q}, \sem{\cdot} \rangle$ be an
     interpretation, where $\mathcal{Q}$ is a quantale.  Then for any $v$,
     \begin{equation*}
       \mathscr{I}\langle v \rangle = \bigoplus_{\pi \in \mathsf{IPaths}(v)} \textit{flatten}(\textit{stack}(\pi))
     \end{equation*}
   \end{theorem}
   \begin{proof}
     The proof proceeds in two steps.  First, we define a \emph{quantified
       word} interpretation $\mathscr{W} = \tuple{\mathcal{W},
       \sem{\cdot}_{\mathcal{W}}}$, for which we can show a coincidence
     result.  Then, we show that any interpretation $\mathscr{I} =
     \tuple{\mathcal{Q}, \sem{\cdot}_{\mathcal{Q}}}$ over a quantale
     ``factors'' through the quantified word interpretation, in the sense that
     there is a quantale homomorphism $h$ such that
     \[\mathscr{I}\tuple{v} =
     h(\mathscr{W}\tuple{v})\]

     We first define the set of quantified words $W$ to be the empty word or
     any word recognized by the following context free grammar:
     \begin{align*}
       W ::= & e & e \in E\\
       | & (\exists x. W) & x \in \Var\\
       | & W W
     \end{align*}
     We then define a (weak) quantale
     \[ \mathcal{W} = \tuple{\powerset{W}, \cup, \mul_{\mathcal{W}}, \star_{\mathcal{W}}, \exists_{\mathcal{W}}, \emptyset, \{\epsilon\}} \]
     where
     \begin{align*}
       P\mul_{\mathcal{W}} Q &= \{ ab : a \in P, b \in Q \}\\
       P^{\star_{\mathcal{W}}} &= \bigcup_{n \in \mathbb{N}} P^n\\
         \exists x. P &= \{ (\exists x. a) : a \in P \}
     \end{align*}
     Note that $\mathcal{W}$ does not satisfy axioms (Q2)-(Q4).  It is for
     this reason that we call $\mathcal{W}$ a weak quantale.  However, the
     axioms (Q2)-(Q4) are irrelevant for this proof.

     We define the quantified word interpretation as $\mathscr{W} =
     \tuple{\mathcal{W}, \sem{\cdot}_{\mathcal{W}}}$, where
     $\sem{e}_{\mathcal{W}}$ is defined as the word consisting of a single
     letter $e$.

     Next, we define a function $h : \mathcal{W} \rightarrow \mathcal{Q}$ by
     \[ h(P) = \bigoplus_{\mathcal{Q}} \{ \tilde{h}(a) : a \in P \} \]
     where
     \begin{align*}
       \tilde{h}(e) &= \sem{e}_{\mathcal{Q}}\\
       \tilde{h}(ab) &= \tilde{h}(a) \mul_{\mathcal{Q}} \tilde{h}(b)\\
       \tilde{h}(\exists x.a) &= \exists_{\mathcal{Q}} x. \tilde{h}(a)
     \end{align*}

     We will use the subscripts $\mathcal{Q}$ and $\mathcal{W}$ on
     $\textit{stack}$ and $\textit{flatten}$ to distinguish between which
     interpretation it is being applied in.

     \begin{lemma} \label{lem:flatten}
       For any interprocedural path $\pi$,
       \[h(\textit{flatten}_{\mathcal{W}}(\textit{stack}_{\mathcal{W}}(\pi))) = 
       \textit{flatten}_{\mathcal{Q}}(\textit{stack}_{\mathcal{Q}}(\pi)) \]
     \end{lemma}
     \begin{proof}
       By induction on $\pi$.
     \end{proof}

     \begin{lemma} \label{lem:factor}
       For any $v$,
       \[ \mathscr{I}\tuple{v} = h(\mathscr{W}\tuple{v}) \]
     \end{lemma}
     \begin{proof}
       Define a relation $\sigma \subseteq \mathcal{W} \times \mathcal{Q}$ as
       the graph of the function $h$.  That is,
       \[ \sigma = \{ \tuple{P,h(P)} : P \in \mathcal{W} \} \]
       It is easy to show that $\sigma$ is a congruence w.r.t. the \QPKA{}
       operations.

       Next, we must show that for any $i$,
       $\sound{\overline{S}_{\mathcal{W}}(i)}{\overline{S}_{\mathcal{Q}}(i)}$.
       We may prove this by induction on the transfinite iteration
       sequence defining $\overline{S}_{\mathcal{W}}$ and
       $\overline{S}_{\mathcal{Q}}$ (using Lemma~\ref{lem:pathexp_rel} for the
       induction step).

       Finally, we may apply Lemma~\ref{lem:analysis_rel} to get
       \[ \sound{\mathscr{W}\tuple{v}}{\mathscr{I}\tuple{v}} \]
       whence $\mathscr{I}\tuple{v} = h(\mathscr{W}\tuple{v})$.
     \end{proof}

     It is easy to see that $\mathscr{W}\tuple{v}$ is exactly the set of
     interprocedural paths to $v$, where calls and returns have been replaced
     by existential quantification of the appropriate variables.  This can be
     expressed formally as
     \[w \in \mathscr{W}\tuple{v} \iff \exists \pi \in \mathsf{IPaths}(v). \{w\} = \textit{flatten}_{\mathcal{W}}(\textit{stack}_{\mathcal{W}}(\pi))\]
     Or equivalently, as
     \[ \mathscr{W}\tuple{v} = \bigcup_{\pi \in \mathsf{IPaths}(v)} \textit{flatten}_{\mathcal{W}}(\textit{stack}_{\mathcal{W}}(\pi)) \]

     We may then conclude
     \begin{align*}
      \mathscr{W}\tuple{v} &= \bigcup_{\pi \in \mathsf{IPaths}(v)} \textit{flatten}_{\mathcal{W}}(\textit{stack}_{\mathcal{W}}(\pi))\\
      h(\mathscr{W}\tuple{v}) &= h\Big(\bigcup_{\pi \in \mathsf{IPaths}(v)} \textit{flatten}_{\mathcal{W}}(\textit{stack}_{\mathcal{W}}(\pi))\Big)\\
      \mathscr{I}\tuple{v} &= h\Big(\bigcup_{\pi \in \mathsf{IPaths}(v)} \textit{flatten}_{\mathcal{W}}(\textit{stack}_{\mathcal{W}}(\pi))\Big) & \text{Lemma~\ref{lem:factor}}\\
      \mathscr{I}\tuple{v} &= \bigoplus_{\pi \in \mathsf{IPaths}(v)} h(\textit{flatten}_{\mathcal{W}}(\textit{stack}_{\mathcal{W}}(\pi))) & \text{Def'n of } h\\
      \mathscr{I}\tuple{v} &= \bigoplus_{\pi \in \mathsf{IPaths}(v)} \textit{flatten}_{\mathcal{Q}}(\textit{stack}_{\mathcal{Q}}(\pi)) & \text{Lemma~\ref{lem:flatten}}
     \end{align*}
   \end{proof}

\section{Experimental results}

We present a detailed comparison of {\sc LRA}, {\sc UFO}, and {\sc InvGen} on
our benchmark suite.  This includes result only for those benchmarks where one
of the answers among the three tools differed from the other two.
\begin{figure}
\begin{tabular}{|l|c|c|c|} \hline

 Benchmark Name           &  {\sc LRA}    & {\sc UFO}     & {\sc InvGen}                \\ \hline \hline

 apache-escape-absolute   & unsafe  & unsafe  & safe                  \\
 apache-get-tag           & safe    & unsafe  & safe                  \\
 gulwani\_cegar2           & safe    & safe    & unsafe                \\
 gulwani\_fig1a            & safe    & safe    & unsafe                \\
 half                     & safe    & unsafe  & unsafe                \\
 heapsort                 & safe    & unsafe  & safe                  \\
 heapsort2                & safe    & unsafe  & safe                  \\
 heapsort3                & safe    & unsafe  & safe                  \\
 id\_trans                 & unsafe  & unsafe  & safe (false negative) \\
 large\_const              & safe    & timeout & safe                  \\
 MADWiFi-encode\_ie\_ok     & safe    & unsafe  & safe                  \\
 nest-if6                 & safe    & safe    & unsafe                \\
 nest-if8                 & safe    & unsafe  & safe                  \\
 nested6                  & safe    & unsafe  & safe                  \\
 nested7                  & safe    & crash   & unsafe                \\
 nested8                  & safe    & timeout & safe                  \\
 nested9                  & safe    & unsafe  & unsafe                \\
 rajamani\_1               & safe    & timeout & safe                  \\
 sendmail-close-angle     & safe    & unsafe  & safe                  \\
 seq-len                  & safe    & safe    & unsafe                \\
 seq2                     & safe    & crash   & safe                  \\
 split                    & unsafe  & timeout & safe                  \\
 svd1                     & safe    & crash   & safe                  \\
 svd2                     & safe    & crash   & safe                  \\
 svd3                     & safe    & crash   & safe                  \\
 svd4                     & unsafe  & crash   & safe                  \\
 up-nd                    & safe    & unsafe  & safe                  \\
 up2                      & safe    & unsafe  & safe                  \\
 up3                      & safe    & unsafe  & safe                  \\
 up4                      & safe    & unsafe  & safe                  \\
 up5                      & safe    & unsafe  & safe                  \\ \hline \hline

 for\_infinite\_loop\_1\_safe & safe    & safe    & crash                 \\
 for\_infinite\_loop\_2\_safe & safe    & safe    & crash                 \\
 sum01\_safe               & safe    & unsafe  & unsafe                \\
 sum02\_safe               & unsafe  & unsafe  & unsafe                \\
 terminator\_02\_safe       & safe    & safe    & crash                 \\
 token\_ring01\_safe        & unsafe  & safe    & --                    \\
 toy\_safe                 & timeout & timeout & --                    \\ \hline \hline 

 Wrong Results (total)     & 6       & 29      & 14                    \\ \hline

\end{tabular}
\caption{Experimental Results of comparing LRA analysis with UFO and InvGen tools on 
the set of InvGen benchmarks, and the set of SVComp benchmarks. ``--'' indicates that the
benchmark has function calls, which not currently supported by {\sc InvGen}. The last row indicates
how many wrong results each tool had (incorrect result, timeout, or crash) over the entire set of benchmarks. }
\end{figure}

\end{document}